\documentclass[11pt]{article}

\usepackage{amsmath,amssymb,fullpage,amsthm,xcolor}
\usepackage[T1]{fontenc} \usepackage{fourier}

\usepackage{cite,color,xspace}
\usepackage{todonotes}

\usepackage{enumitem}

\usepackage{hyperref}

\usepackage{tikz,ifthen}

\newtheorem{thm}{Theorem}
\newtheorem{lemma}[thm]{Lemma}
\newtheorem{cor}[thm]{Corollary}
\newtheorem{prop}[thm]{Proposition}

\newtheorem{defn}[thm]{Definition}

\newtheorem{claim}[thm]{Claim}

\newcommand{\vx}{\mathbf{x}}

\newcommand{\R}{\mathbb{R}}

\newcommand{\ED}{\mathrm{ED}}

\newcommand{\cK}{K}
\newcommand{\cT}{\mathcal{T}}

\newcommand{\ST}{\mathrm{ST}}
\newcommand{\SD}{\mathrm{DISJ}}
\newcommand{\AND}{\mathrm{AND}}
\newcommand{\OR}{\mathrm{OR}}

\newcommand{\eps}{\epsilon}

\newcommand{\MCF}{\mathrm{MCF}}

\newcommand{\ml}[1]{\todo[inline, color=cyan]{#1 \hfill --Mike}}

\newcommand{\eqdef}{\stackrel{\text{def}}{=}}

\newcommand{\timeG}[1]{G^{(#1)}}

\newcommand{\CC}{R^{(2)}}
\newcommand{\CCo}{R^{(\to)}}

\newcommand{\tO}{\tilde{O}}
\newcommand{\tOm}{\tilde{\Omega}}
\newcommand{\tT}{\tilde{\Theta}}

\newcommand{\tOmega}{\tilde{\Omega}}
\newcommand{\cP}{\mathcal{P}}
\newcommand{\cQ}{\mathcal{Q}}
\newcommand{\const}{4}
\usepackage{amsthm,amsmath,amssymb}
\usepackage{algorithm,algorithmic}

\newcommand{\cA}{A}
\newcommand{\cB}{B}

\newcommand{\vy}{\mathbf{y}}

\newcommand{\set}[1]{\left\{#1\right\}}

\newcommand{\alphaCMG}{\alpha_{\mathrm{CMG}}}

\renewcommand{\tilde}{\widetilde}

\newcommand{\tauroute}{\tau_{\mathrm{route}}}

\newcommand{\ceil}[1]{{\left\lceil#1\right\rceil}}

\newcommand{\EQ}{\mathrm{EQ}}

\newcommand{\ORSD}{\mathrm{OR\operatorname{-}DISJ}}
\newcommand{\ANDSD}{\mathrm{AND\operatorname{-}DISJ}}
\newcommand{\TRB}{\mathrm{TRIBES}}

\newcommand{\SYM}{\mathrm{SYMM}}

\newcommand{\ORSDTWO}{\mathrm{OR\operatorname{-}DISJ\operatorname{-}2PARTY}}
\newcommand{\ANDSDTWO}{\mathrm{AND\operatorname{-}DISJ\operatorname{-}2PARTY}}
\title{Tight Network Topology Dependent Bounds on Rounds of Communication\footnote{AC's research is partially supported by a
Ramanujan fellowship of the DST. ML's research is supported in part by NSF-CCF1526771. SL's research is partly supported by NSF grant CCF-1566356. AR's research is supported in part by NSF grant CCF-1319402.}}

\author{\textsc{Arkadev Chattopadhyay}\footnotemark[4] \and \textsc{Michael Langberg}\footnotemark[2] \and \textsc{Shi Li}\footnotemark[3]  \and \textsc{Atri Rudra}\footnotemark[3]}
\date{\footnotemark[4]~~School of Technology and Computer Science,\\
Tata Institute of Fundamental Research\\
\texttt{arkadev.c@tifr.res.in}\\
\vspace*{2mm}
\footnotemark[2]~~Department of Electrical Engineering,\\
University at Buffalo, SUNY\\
\texttt{mikel@buffalo.edu}\\
\vspace*{2mm}
\footnotemark[3]~~Department of Computer Science and Engineering,\\
University at Buffalo, SUNY\\
\texttt{\{shil,atri\}@buffalo.edu}
}
\begin{document}
\maketitle

\setcounter{page}{0}
\thispagestyle{empty}

\begin{abstract}
We prove tight network topology dependent bounds on the round complexity of computing well studied $k$-party functions such as set disjointness and element distinctness. Unlike the usual case in the CONGEST model in distributed computing, we fix the function and then vary the underlying network topology. This complements the recent such results on total communication that have received some attention. We also present some applications to distributed graph computation problems.

Our main contribution is a proof technique that allows us to reduce the problem on a general graph topology to a relevant two-party communication complexity problem. However, unlike many previous works that also used the same high level strategy, we do {\em not} reason about a two-party communication problem that is induced by a cut in the graph. To `stitch' back the various lower bounds from the two party communication problems, we use the notion of timed graph that has seen prior use in network coding. Our reductions use some tools from Steiner tree packing and multi-commodity flow problems that have a delay constraint.
\end{abstract}

\newpage

\section{Introduction}

In this paper, we prove bounds on the number of rounds needed to compute a given function in a distributed manner.  In our paper a problem is a tuple $(f,G,\cK)$, where $G=(V,E)$ is the underlying communication graph (which is assumed to be {\em undirected}), $\cK\subseteq V$ is a set of $k \eqdef |\cK|$ terminals (or players), and we are interested in computing the function $f:\left(\{0,1\}^n\right)^{\cK}\to\{0,1\}$: i.e. all terminals in $\cK$ need to know the final answer after the protocol is done.\footnote{It turns out if one terminal knows the answer then it can send the answer to all others via a simple Steiner tree based protocol whose cost is dominated by all our bounds.}
Unless stated otherwise, the $k$ inputs are assigned in worst-case manner to the terminals in $\cK$.

All communication in a protocol is point-to-point (as opposed to the broadcast mode of communication) and a bit transmitted over an edge $e=(u,v)$ is private to $u$ and $v$. Further, we assume a synchronous model and in each {\em round}, each node $u\in V$ sends a (potentially different) bit\footnote{By creating parallel edges, our results extend to the case where in each round each edge $e\in E$ can send $c_e$ bits. However, for notational simplicity we will only consider the case of $c_e=1$ in this paper; that is, $G$ is a simple graph.} to each of its neighbors. We assume the two directions of an edge $(u, v)$ can be used simultaneously. We will further assume that the protocols have full knowledge of $G$ and all nodes (for randomized protocols) use public randomness.\footnote{Since all parties know the entire topology $G$, then one can generalize Newman's argument~\cite{newman} to our setting. In our case, in the private randomness protocol corresponding to the original public randomness protocol, one party will needs to send $O(\log{nk})$ bits to all other $k-1$ parties. This can be accomplished by a simple Steiner tree protocol, whose cost can be absorbed in all of our bounds. We will stick with public randomness since it makes the description of our protocols easier to follow.}
In this paper, we are interested in the round complexity: i.e. the total number of rounds needed by a protocol to compute the output. Note that this notion corresponds to the time taken by the distributed protocol to compute the answer.
Given a problem $\mathcal{P}$ we will use $R_{\eps}(\mathcal{P})$ to denote the minimum number of rounds needed for the worst-case input of any randomized protocol that errs on all input with probability at most $\eps$. WLOG for randomized protocols one can assume that $\eps=1/3$ and we will in most cases refer to $R_{1/3}(\mathcal{P})$ by just $R(\mathcal{P})$. Note that $R_0(\mathcal{P})$ denotes the {\em deterministic} round complexity.\footnote{We note that in communication complexity literature, $R_0$ is used to denote the zero-error randomized communication complexity but we use this convention since it makes our theorem statements cleaner.} To simplify our presentation we will ignore in our bounds poly-logarithmic factors in both the size of $G$ and $n$. In particular, we will use the notation $\tO(\cdot), \tOm(\cdot)$ and $\tT(\cdot)$ to denote the usual asymptotic notation that ignore poly-log factors (in size of $G$ and $n$).

Our model above is very similar to the well studied CONGEST model in distributed computing~\cite{peleg-book} with the following differences. First, for proving upper bounds on the CONGEST model, it is assumed that a node in $V$ only knows about its neighbors while in our setup we assume that the protocol knows the structure of $G$. This makes our lower bounds potentially stronger (though this makes our upper bounds weaker than the distributed protocol bounds in the CONGEST model). 
Second, typically in the distributed computing literature the function $f$ itself depends on the underlying network $G$ (e.g. check if a given subgraph of $G$ is a spanning tree~\cite{das-sarma}) while in our setup the function $f$ is independent of the network topology $G$. This assumption makes sense in the current state of affairs where many such functions are computed in a distributed manner over the same network. Recent works (including those of Drucker at al.~\cite{DKO13} and Klauck et al.~\cite{klauck}) have proved  bounds for functions for the special case where $G$ is the complete graph. Finally, in most of the existing work it is assumed that $\cK=V$, while we consider the more general case when $\cK\subseteq V$. This more general case makes sense e.g. in a data warehouse where any given function that needs to be computed could only depend on inputs that are stored at some subset of the servers.

Recently there has been work that deals with the graph communication model as above but instead of minimizing the round complexity, these results are for the case of minimizing the {\em total communication} of the protocols.  (We note that the total communication corresponds to the {\em message complexity} of distributed protocols.) Most of the work in this area has been for specific classes of $G$. For example, the early work of Tiwari~\cite{T87} considered deterministic total communication complexity on cases of $G$ being a path, grid or ring graph. There has been a recent surge of interest for proving lower bounds on total communication for the case when $G$ is a star~\cite{PVZ12,WZ12,WZ13,BEOPV13,WZ14,CM15}. 
This work was generalized to arbitrary topology by Chattopadhyay et al.~\cite{CRR14} who proved tight bounds for certain functions for {\em all} network topologies. A followup work extended the results to some more functions~\cite{CR15}.

Both of these strands of work (on round complexity and total communication) coincide for the special case when $G$ is just an edge. Note that in this case we have two players and the model coincides with the very well studied model of two-party communication complexity introduced by Yao~\cite{Yao79}, which has proved to be an extremely worthwhile model to study with applications in diverse areas of theoretical computer science.


Given the importance of round complexity in distributed computing, it is natural to ask 
\vspace*{-5pt}
\begin{quote}
\textit{Can we prove tight topology sensitive bounds 
	for round complexity?}
\end{quote}
\vspace*{-5pt}

We would like to point out that optimal protocols for total communication need not be optimal for round complexity and vice-versa.  To see this, consider the case where $G$ contains two terminals $\set{a, b}$ and many parallel edge-disjoint paths between $a$ and $b$: there is one path of length $1$, and $\sqrt{n}$ paths of length $\sqrt{n}$. $a$ receives $n$ bits and wants to send those $n$ bits to $b$. The optimal protocol for total communication would be to send the $n$ bits on the length-$1$ path, which has $O(n)$ total communication but takes $\Omega(n)$ rounds. On the other hand, an (almost) optimal protocol in terms of number of rounds would be splitting the $n$ bits into $\sqrt{n}$ blocks of $\sqrt{n}$ bits and send each block using one of the $\sqrt{n}$ paths of length $\sqrt{n}$. This protocol has round complexity $O(\sqrt{n})$ but total communication $\Omega(n \sqrt{n})$.

In this work, we prove tight bounds on round complexity for several families of functions.
We believe that the proof techniques presented, and not just the concrete results, are of independent interest.
Pretty much all of the previous work in the total communication regime proved their lower bounds via two steps. The first step was to `divide' up the problem into a bunch of two party communication complexity problems. The second step is to `stitch' together the lower bounds for these two party communication problems. Our proofs also have the same two step structure but our implementations of both these steps are very different. The first step in previous works is implemented by constructing a family of cuts and then considering the two-party problem induced on each cut. In our proofs, we consider a more general set of edges $E'$ (which might not form a cut) and then simulate our original protocol on $G$ projected down to $E'$ via a two-party communication protocol. The second step in previous work used a common hard distribution across all chosen cuts and then used linearity of expectation to `add up' the lower bounds. In contrast, we use the notion of a {\em timed graph} (that is independent of the hard distributions) so that we can use different hard distributions for the different two party communication problems to deduce something about the {\em same} timed graph. The stitching then occurs by proving various `gluing' results on Steiner tree packing and multi-commodity flow problems on graphs. This difference allows us to prove lower bounds for both randomized and deterministic protocols with the same proof while e.g. the results of~\cite{CRR14} could not prove tight deterministic bounds for functions whose zero-error randomized complexity is much smaller than its deterministic complexity.

\subsection{Overview of our results}
We prove our bounds for two classes of functions, as in \cite{CR15}.
Roughly speaking in the total communication setting, one class has an optimal protocol that combines inputs upwards on a Steiner tree and in the second class of problems the optimal protocol involves all players sending their inputs to a designated node. Next, we define two functions that are representatives of these two classes. (See Theorems~\ref{thm:sd--gen-lb-main} and~\ref{thm:ed-lb-main} for the exact definitions of these two classes.)

We start our overview with the well studied $k$-party set disjointness problem, which is defined as follows. Each player $u\in\cK$ gets a string $\vx_u\in\{0,1\}^n$ (which can be thought of as a subset of $[n]\eqdef \{1,2,\dots,n\}$) and the output is
\[\SD_{\cK, n}\big(\{\vx_u\}_{u\in\cK}\big)=\bigvee_{i\in [n]} \bigwedge_{u\in \cK} \vx_u[i],\]
i.e. the output is $1$ if and only if all the $k$ sets have an element in common.

For $\SD_{\cK, n}$ as shown in~\cite{CR15}, the optimal protocol (up to poly-logarithmic factors) for total communication is to first compute the minimum Steiner tree on $G$ with $\cK$ as the set of terminals and then to compute the intersection of the $k$ sets in a bottom-up fashion. For the round complexity, it seems natural to try this scheme in `parallel': i.e. try to {\em pack} as many edge disjoint Steiner trees of small diameter as possible and to compute the set intersection on appropriate parts of the universe $[n]$ up the trees in parallel. It turns out that this is indeed the optimal protocol for round complexity. 
We prove the following result (where $\ST(G,\cK,\Delta)$ denotes the optimal value of Steiner tree packing with terminals $\cK$ and diameter $\Delta$ in $G$; formal definition appears in Section~\ref{sec:prelims}):
\begin{thm}
\label{thm:k-SD}
For any graph $G$ and subset of players $\cK$, we have for every $\eps\ge 0$
\[R_{\eps}(\SD_{\cK, n},G,\cK)  = \tT\left( \min_{\Delta\in [|V|]} \left(\frac{n}{\ST(G,\cK,\Delta)}+\Delta\right)\right).\]
\end{thm}

The other function is the element distinctness problem (shortened to $\ED$), which is defined as follows. Each player $u\in\cK$ gets a string $\vx_u\in\{0,1\}^n$ (which can be thought as a number in $[0,2^n-1]$) and the output is
\[\ED_{\cK, n}\big(\{\vx_u\}_{u\in\cK}\big)=\bigwedge_{u\neq v\in \cK} \vx_u\neq \vx_v.\]

For $\ED_{\cK, n}$ as shown in~\cite{CRR14}, the optimal randomized protocol for total communication is for the $k$ players to send the hash of their inputs to the median node w.r.t. $\cK$ in $G$. 
A natural protocol would be to run a multi-commodity flow problem where the demands correspond to each of the $k$ players sending their bits to the median node. However, it turns out that this is {\em not} optimal for round complexity. Intuitively the main reason this fails is because the median node has too much incoming flow. The next natural idea would be to somehow have a different multi-commodity flow problem where each node has a `balanced load'. Indeed we are able to show this to be possible by using a small circuit for $\ED_{\cK, n}$ as our guide.  Let $\tau_{\MCF}(G,\cK,n')$ denote the smallest number of rounds $\tau$ needed to simultaneously route $n'/k$ units flow from $u$ to $v$ for every $u,v\in\cK$. Then we show that
\begin{thm}
\label{thm:ED}
For any $G$ and $\cK$, we have for any constant $\eps>0$
\[R_{\eps}(\ED_{\cK, n},G,\cK) =\tT\left(\tau_{\MCF}(G,\cK,1)\right)\]
and
\[R_0(\ED_{\cK, n},G,\cK) =\tT\left(\tau_{\MCF}(G,\cK,n)\right).\]
\end{thm}

In particular, we generalize the construction in Drucker et al.~\cite{DKO13} to show how to convert any bounded fan-in and fan-out circuit for any function $f$ into a protocol for $f$. Drucker et al. proved such a result for the special case of $G$ being the complete graph.\footnote{However,~\cite{DKO13} do not lose any $\tO(1)$ factors like we do.} More specifically, we show that
\begin{lemma}
\label{lem:ckt-protocol}
Let $f:\left(\{0,1\}^n\right)^k\to \{0,1\}$ have a  circuit with constant fan-in and constant fan-out gates and depth $d$. Further, each level $i\in [d]$ has $s_i$ gates in it (and let $s=\sum_{i=1}^d s_i$). Then
\[R_0(f,G,\cK) \le \sum_{i=1}^d \tO\left(\tau_{\MCF}\left(G,\cK,\frac{s_i}{k}\right)\right).\]
Finally, we can upper bound the above by $\tO\left(d\cdot \tau_{\MCF}\left(G,\cK,\frac{s}{k}\right)\right)$ as well as $\tO\left(\frac{s}{k}\cdot \tau_{\MCF}\left(G,\cK,1\right)\right)$.
\end{lemma}

Like the results of Drucker et al., this connection implies a barrier to proving quantitatively better lower bounds. In particular, if we could exhibit an explicit function $f$ that we could prove requires asymptotically larger number of rounds than $\tO\left(d\cdot \tau_{\MCF}\left(G,\cK,n\right)\right)$, then we would have shown a super-linear size lower bound for circuits computing $f$ with depth $d$. This extends the results in~\cite{DKO13}, which showed that for the clique topology proving any super-constant lower bound on rounds for an explicit function $f$ will imply a corresponding (new) circuit lower bounds.

We also apply our general lower bounds to prove lower bounds for the following distributed graph problems. In these problems each player $u\in\cK$ gets a graph $H_u$ as input and the goal is to check if the overall graph $H=\bigcup_{u\in \cK}H_u$ has certain properties. In particular, we consider the following four problems that check if $H$  (i) is connected, (ii) contains a triangle, (iii) is acyclic, (iv) is connected. 
We show a lower bound of $\tOm\left(\tau_{\MCF}\left(G,\cK,\frac{|V(H)|+|E(H)|}{k}\right)\right)$. 
Our lower bounds extend some of the lower bounds in~\cite{klauck} to general topologies. In particular, we generalize the lower bounds for connectivity to general topologies while~\cite{klauck} does not provide any lower bounds for the triangle detection problem. (However, we note that~\cite{klauck}'s lower bound for connectivity holds for random distribution of $H$ while our lower bounds assume worst-case distribution. So our lower bounds are proven in a weaker setting.)
We also show, by simple adaptation of upper bounds in~\cite{DKO13,klauck}, that as long as $H$ is large, these bounds are also tight. 

Finally, we highlight a technical result that we believe is of independent interest.  We will use $\CC_{\eps}(f)$ to denote the randomized round complexity for the two party case (we assume Alice and Bob can send a bit to each other simultaneously in each round), where we allow Alice and Bob to have inputs of different sizes. We will also need to consider $\CCo_{\eps}(f)$ for the {\em one-way} round complexity where (say) Alice (or Bob) sends a single message to Bob (Alice resp.) and Bob (Alice resp.) computes the answer based solely on the single message he (she resp.) received from Alice (Bob resp.) as well as his (her resp.) input.
Let $\tauroute(G,\{u,v\},n')$ denote the minimum number of rounds in which $u$ can route $n'$ bits to $v$ in $G$; since $G$ is undirected, this is the same as the minimum number of rounds in which $v$ can route $n'$ bits to $u$ in $G$; thus the notation is well-defined.

\begin{thm}
	\label{thm:2pSD}
	For any function $f:\{0,1\}^n\times \{0,1\}^n\to \{0,1\}$, and any graph $G$ we have that
	\[\tauroute(G,\{a,b\},\CC_{\eps}(f)) \le 4 {R_{\eps}(f,G,\{a,b\}}).\]
\end{thm}

Notice the above inequality implies that $\frac{\tauroute(G,\{a,b\},\CCo_{\eps}(f))}{{R_{\eps}(f,G,\{a,b\}}) } \le 4  \ceil{\frac{\CCo_{\eps}(f)}{\CC_{\eps}(f)}}$, by Claim~\ref{claim:tau-sub-additive} (stated in Section~\ref{sec:prelims}).
The technical result implies that when using the obvious one-way  communication algorithm to solve $f$ for the case of $k=2$, the penalty we incur on {\em any} graph is no worse than a constant factor when used on the case when $G$ is just the edge $(u,w)$ (i.e. the traditional two party communication complexity setting).

\subsection{Overview of our proof techniques.} 

We now present an overview of our proof techniques specialized to the case of $\SD_{\cK, n}$ and $\ED_{\cK, n}$. To begin with we will assume that $n$ is much larger than the size of $G$. In this case the common way to prove a lower bound is via the so called {\em communication bottleneck} argument: if $m$ bits have to transmitted over a cut $C$ with $\delta(C)$ crossing edges, then the obvious lower bound on round complexity is $\frac{m}{|\delta(C)|}$. 
	
	We begin with $\SD_{\cK, n}$.  Using the communication bottleneck argument and the linear lower bound on the two-party communication complexity of $\SD$~\cite{Razborov92}, we get a lower bound of $\Omega\left(\frac{n}{\lambda_{\cK}(G)}\right)$ for $\SD_{\cK, n}$ (where $\lambda_{\cK}(G)$ is the size of the min cut separating $\cK$). For the upper bound, we invoke a result of Lau\cite{Lau07} to argue that we can pack $T=\Omega\left(\lambda_{\cK}(G)\right)$ many {\em edge disjoint} Steiner trees in $G$ with terminals $\cK$. If $n$ is large enough then the number of rounds needed is $O\left(\frac nT\right) = O\left(\frac{n}{\lambda_{\cK}(G)}\right)$, giving a tight bound. 
	
	We now consider the case of $\ED_{\cK, n}$ with large enough $n$.  The communication bottleneck argument along with known relations between multi-commodity flow and sparsity of a graph~\cite{LR99,LLR95} gives us a randomized lower bound of $\tOm(\tau_{\MCF}(G,\cK,1))$.
	The trivial protocol of all players sending (the hash of their) inputs to one player $r$ gives us an upper bound of $\tO(\tau_{\MCF}(G,\cK,k))$. The mis-match is because in this case $r$ has an incoming flow of $\Omega(k)$. To avoid this we adapt the argument in~\cite{DKO13} to have $\tO(1)$ phases, where each phase is a more `balanced' multi-commodity flow problem that can be solved in $\tO(\tau_{\MCF}(G,\cK,1))$ number of rounds. The flow problems in these phases are guided by a small circuit that computes the $\ED$ function, as in  Lemma~\ref{lem:ckt-protocol}.

	It turns out that for the above results for $\ED_{\cK, n}$ to hold (for randomized complexity), $n$ has to be {\em exponentially} larger than the size of $G$, which is not ideal (and something we would like to avoid assuming). It turns out that the reason we need $n$ to be large enough for the above arguments is that the results for Steiner tree packing~\cite{Lau07} and those of multicommodity flow~\cite{LR99,LLR95} are only proved without any constraints on the diameter of the Steiner trees and the dilation (i.e. the length of the longest flow) of the multicommodity flow. In our arguments, we take both of these factors into account. In particular, for the upper bounds we simply `pick' the best Steiner tree packing and multicommodity flows based on delay constraints.  
	
	However, the Steiner-tree packing result of \cite{Lau07} and the flow-cut-gap results of \cite{LR99,LLR95} break down if we impose the diameter constraints on the Steiner-trees, or the dilation constraint on the multi-commodity flow.  We need to use other techniques to handle these constraints. For the Steiner-tree packing problem with diameter constraints, we apply the techniques for bi-criteria network design in\cite{MRS98}. In particular,  \cite{MRS98} gave an $\big(O(\log |V|), O(\log |V|)\big)$-approximation algorithm for the bounded diameter minimum Steiner tree problem, where the first $O(\log |V|)$ factor is for the violation of the diameter constraint and the second $O(\log |V|)$ factor is for the cost of the Steiner tree. Using the duality between maximizing Steiner-tree packing and minimizing cost of a Steiner tree, we are able to give a good Steiner-tree packing that approximately satisfies the diameter constraint.   
	
	For the multi-commodity flow problem with the dilation constraint, we could not give a good bi-criteria approximation for all demand functions.  However, for certain demand functions (including the demand function corresponding to $\ED$), we can apply  the cut-matching-game technique in \cite{KRV09} for constructing expanders. For these demand functions, for every equal partition $(\cA, \cB)$ of $\cK$, we can route many matchings between $\cA$ and $\cB$ using short paths. The cut-matching game technique allows us to find a small-congestion ``embedding'' of an expander in $G$ using these paths, which along with the properties of the expander allow us to route the demand appropriately. 
	
	A crucial ingredient in our lower bound proofs is the notion of a {\rm timed graph}, which was introduced in the context of network coding in the study of cyclic networks \cite{ahlswede2000network}. Timed graphs have found several applications in the network coding literature including in the study of time constrained network communication, memory constrained network communication, and {\em gossip} protocols, e.g., \cite{haeupler2011optimality,wang2014sending,chekuri2015delay}.
	Informally, the $\tau$-timed version graph of $G$, which we denote by $G^{(\tau)}$ is a graph with $\tau+1$ layers (with $\tau+1$ copies of $V$) with the edge set of $E$ repeated between the $i$ and $(i+1)$th layer (for $0\le i<\tau$). The crucial property of the timed graph is that there exists protocol $\Pi$ over $G$ with round complexity $\tau$ if and only if there exists a protocol on $G^{(\tau)}$ where each edge is used at most once.
	
Finally, we present an overview of how we argue for the presence of good Steiner tree packing and good multi commodity flow in $G^{(\tau)}$, where $\tau$ is the number of rounds taken by the optimal protocol. The obvious thing to try here would be to again appeal to two party communication complexity lower bounds on cuts on $G^{(\tau)}$ itself and then appeal to known results relating Steiner tree packing and multicommodity flow on directed graphs to the corresponding functions on cuts. There are two issues. First, for Steiner tree packing and multi commodity flow the integrality gap for general directed graphs are either unknown or unbounded (which is not helpful). We get around this issue by explicitly using the fact that $G^{(\tau)}$ is a special graph: i.e. a timed graph of an {\em undirected} graph. The second issue is that directly applying the two party communication complexity lower bounds across a cut in $G^{(\tau)}$ is typically not enough since these only imply a lower bound on number of crossing edges in {\em both} directions across a cut in $G^{(\tau)}$, while our argument require a lower bound on the number of edges going from `left' to `right' in a cut. We address this issue by invoking the two party communication complexity not across a cut in $G^{(\tau)}$ but invoking it on a carefully chosen subgraph of $G^{(\tau)}$.

\subsection*{Organization of the paper}

We start off with some preliminaries in Section~\ref{sec:prelims}. We prove Theorem~\ref{thm:2pSD} in Section~\ref{sec:k=2}. We prove Theorem~\ref{thm:k-SD} (and its generalization Theorem~\ref{thm:sd--gen-lb-main}) in Section~\ref{sec:sd}. We prove Theorem~\ref{thm:ED} (and its generalization Theorem~\ref{thm:ed-lb-main}) in Section~\ref{sec:ed}. Finally, we present our bounds for distributed graph problems in Section~\ref{app:apps}.

\section{Preliminaries}
\label{sec:prelims}

\subsection{Notations} 
Let $f: \big(\{0, 1\}^n\big)^{\cK} \to \{0, 1\}$ be a function, and $\cA, \cB \subseteq \cK$ be two disjoint non-empty sets and $\tilde \vx \in \big(\{0, 1\}^n\big)^{\cK\setminus (\cA \cup \cB)}$. Define $f_{\cA, \cB, \tilde \vx} : \big(\{0, 1\}^n\big)^{\cA} \times \big(\{0, 1\}^n\big)^{\cB} \to \{0, 1\}$ to be the function such that for every $\vx_{\cA} \in \big(\{0, 1\}^n\big)^{\cA}$  and $\vx_{\cB} \in \big(\{0, 1\}^n\big)^{\cB}$, we have $f_{\cA, \cB, \tilde \vx}(\vx_{\cA}, \vx_{\cB}) = f(\tilde \vx \circ \vx_{\cA} \circ \vx_{\cB} )$, where $\tilde \vx \circ \vx_{\cA} \circ \vx_{\cB}$ denotes the vector $\vy \in \left(\{0, 1\}^n\right)^\cK$ such that $\vy[v] = \vx_\cA[v]$ if $v \in \cA$, $\vy[v] = \vx_\cB[v]$ if $v \in \cB$ and $\vy[v] = \tilde \vx[v]$ if $v \in \cK \setminus (\cA \cup \cB)$. For every pair $A, B \subseteq K$ of disjoint non-empty sets, we use $G_{\cA, \cB}$ to denote the graph $G$ with vertices in $\cA$ identified, and vertices in $\cB$ identified.  We shall use $v_\cA$ and $v_\cB$ to denote the two new vertices in $G_{\cA, \cB}$.

\subsection{The timed graph}

We now define a graph related to $G$ that will be crucial in our arguments. Given an integer $\tau\ge 1$, we define a directed and layered graph $\timeG{\tau}=(V_{\tau},E_{\tau})$, where
\[V_{\tau}=V\times [0,\tau],\]
and $E_{\tau}$ is defined as follows. For every $(u,v)\in E$, we have the following edges in $E_{\tau}$:
\[\{ ((u,i),(v,i+1))| 0\le i<\tau\} \cup \{ ((v,i),(u,i+1))| 0\le i<\tau\}.\]
Finally we add to $E_{\tau}$ infinitely many parallel edges $( (u,i), (u,i+1))$ for every $u\in V$ and $0\le i<\tau$ that we call {\em memory edges}.\footnote{We only need large enough memory edges. However, we choose to say there are infinite of them to avoid having to specify the exact number of such edges.}

A useful property of $G^{(\tau)}$ is that given a protocol with congestion $c$ on $G^{(\tau)}$ one can easily construct a protocol with no congestion on $G^{(c\cdot \tau)}$, i.e. we can get a valid protocol on $G$ with delay $c\tau$. This makes our arguments simpler since in the Steiner tree packing and multi-commodity flow solutions we can tolerate $\tO(1)$ congestion.

\subsection{Graph background}

We recall two graph problems that have been studied extensively and will be crucial in our analysis.

\paragraph{Steiner Tree Packing.} We begin with the problem of Steiner tree packing. Given the graph $G$ and set of terminals $\cK$, we call a tree $T$ a Steiner tree if it connects all vertices in $\cK$ only using edges in $G$. We consider the (fractional) Steiner tree packing problem, which we will represent by the following well-studied LP. In particular, we would be interested in Steiner trees with diameter (between any two terminals) of $\Delta$-- let $\cT_{\Delta,\cK}$ denote the set of all such Steiner trees.
\[\max \sum_{T\in\cT_{\Delta,\cK}} z_T \qquad \text{s.t.} \qquad \sum_{T\ni e} z_T \le 1 \text{ for every }e\in E, \qquad z_T\ge 0, \forall T \in \cT_{\Delta, \cK}.\]

Let $\ST(G,\cK,\Delta)$ denote the optimal value of the above LP.

%

\paragraph{Multi-commodity flow.} We will also use the well-studied multi-commodity flow problem. A demand function $D$ is some vector in $\R_{\geq 0}^{\cK \times \cK}$. In this demand, we need to send $D_{u, v}$ units of flow from $u$ to $v$ for every $u \in \cK, v \in \cK$.  Since we are interested in the minimum number of rounds to route the demand function $D$, it is convenient to view the demand as directed and not necessarily symmetric: for every $u, v \in \cK$, we need to send $D_{u, v}$ units of flow from $u$ to $v$ and $D_{v, u}$ units of flow from $v$ to $u$, where $D_{u,v}$ and $D_{v, u}$ may be different.   In the problem, we assume that in each round for every edge $(u, v) \in E$, we can send at most 1 unit flow from $u$ to $v$ and at most 1 unit flow from $v$ to $u$.

\begin{defn}
	For any real number $n' > 0$, we say a demand function $D  \in \R_{\geq 0}^{\cK \times \cK}$ is $n'$-bounded, if for every $u \in \cK$, we have $\sum_{v \in \cK} D_{u,v}\le n'$ and $\sum_{v \in \cK} D_{v, u}\le n'$.
\end{defn}

\begin{defn}
	For every $n' > 0$, let $\tau_{\MCF}(G, \cK, n')$ be the minimum number of rounds $\tau$ such that we can simultaneously send $n'/k$ units flow from $u$ to $v$ in $G$, for every $u, v \in \cK$.  For every $a, b \in V$, let $\tauroute(G, \{a, b\}, n')$ denote the minimum number of rounds $\tau$ such that $a$ can send $n'$ units flow to $b$.
\end{defn}
In other words, $\tau_{\MCF}(G, \cK, n')$ is minimum number of rounds to route $D$, for the function $D$ with $D_{u, v} = n'/k$ for every $u, v \in \cK$. Note that
\begin{prop}
\label{prop:mcf-clique}
When $G$ is a clique on $k$ vertices, we have
\[\tau_{\MCF}(G, \cK, n') =\left \lceil \frac{n'}{k}\right\rceil.\]
\end{prop}

 We first note a simple property of $\tau_{\MCF}$ and $\tauroute$.
\begin{claim}
	\label{claim:tau-sub-additive}
	For every $n'>0$ and $n''>0$ and $a, b \in V$, we have $\tau_{\MCF}(G, \cK, n'') \leq \ceil{\frac{n''}{n'}}\tau_{\MCF}(G, \cK, n')$, and $\tauroute(G, \{a, b\}, n'') \leq \ceil{\frac{n''}{n'}}\tauroute(G, \{a, b\}, n')$.
\end{claim}

Next, we note that the definition of $\tau_{\MCF}$ is enough to capture all $n'$-bounded demands. 
\begin{lemma}
	\label{lem:bounded-dem}
	For every $n'$-bounded demand $D$ over $\cK$, we can route $D$ with $2\tau_{\MCF}(G, \cK, n')$ rounds.
\end{lemma}
\begin{proof}
        Without loss of generality, we assume for every $u \in \cK$ we have $\sum_{v \in \cK}D_{u, v} = n'$ and $\sum_{v \in \cK}D_{v, u} = n'$. We route the demand $D$ in 2 stages, each with delay $\tau_{\MCF}(G, \cK, n')$.   We color the commodities by their destinations.  So at the beginning, there are $D_{u, v'}$ units of commodity of color $v'$ at $u$, for every $u, v' \in \cK$.  In the first stage, we send $n'/k$ units of commodity from every $u \in \cK$ to every $v \in \cK$, such that the commodity of each color is split evenly: $v$ is getting $D_{u, v'}/k$ units of commodity of color $v'$ from $u$, for every color $v' \in \cK$. Thus, at the end of the first stage, every vertex $u$ has $n'/k$ units commodity of each color $v'$.  Then, in the second stage, we send the commodity of each color $v'$ to $v'$.  Notice that in each of the two stages, we are sending $n'/k$ units of flow from every $u$ to every $v$ and thus the delay is $\tau_{\MCF}(G, \cK, n')$; so overall the delay is $2\tau_{\MCF}(G, \cK, n')$.
\end{proof}

\paragraph{Facts about expanders.}

Given a graph $H = (V_H, E_H)$, the expansion of $H$ is defined as
                \begin{equation*}
                        \Phi(H):= \min_{S \subseteq V_H: |S| \leq |V_H|/2}\frac{\big|E_H(S, V_H \setminus S)\big|}{|S|},
                \end{equation*}
                where $E_H(S, V_H \setminus S)$ is the set of edges in $E_H$ with one endpoint in $S$ and the other endpoint in $V_H \setminus S$.  We say a graph is an $\alpha$-expander if its expansion is at least $\alpha$.

                Let $H$ be a $d$-regular graph and $A$ be the adjacency matrix of $H$: for every $u, v \in \cK$, $A_{u, v}$ is the number of edges between $u$ and $v$ in $X$.  Since $A$ is symmetric, it has $n$ real eigenvalues. The largest eigenvalue of $A$ is $\lambda_1 = d$. Let $\lambda_2 \leq d$ be the second largest eigenvalue of $A$.  Cheeger's inequality relates $\lambda_2$ and the expansion $\Phi(H)$ of $H$.
                \begin{thm}[Cheeger's Inequality]
                        $\frac{d-\lambda_2}{2}\leq \Phi(H) \leq \sqrt{2d(d-\lambda_2)}$.
                \end{thm}

                We are interested in the following lazy random walk on a $d$-regular graph $H$. We start from an initial vertex $v \in V_H$, chosen randomly according to some initial distribution $q$.  In each step, with probability $1/2$, we stay at the current vertex; with the remaining $1/2$ probability, we move to a randomly selected neighbor of the current vertex. Then,  $(I+A/d)/2$ is the transition matrix of the lazy random walk, where $I$ is the identity matrix.
                The following theorem says that the mixing time of the lazy random walk on an expander is small.
                \begin{thm}[Lazy random walk on expanders] \label{thm:random-walk}
                        Let $H = (V_H, E_H)$ be a $d$-regular graph with $|V_H| = N_H$, $A$ be its adjacency matrix and $\lambda_2$ be the second largest eigenvalue of $A$.  Let $\mu=\big(\mu_v = \frac{1}{N_H}\big)_{v \in V_H}$ be the uniform distribution over vertices in $V_H$. For any initial distribution $q \in [0, 1]^{V_H}$ over $V_H$ and integer $T \geq 0$, we have
                        \begin{equation*}
                                \left\|\left(\frac{I+A/d}{2}\right)^Tq-\mu\right\|_1\le\sqrt{N_H}\left(\frac{1+\lambda_2/d}{2}\right)^T.
                        \end{equation*}
                \end{thm}

\subsection{Circuits} 

We will consider circuits that compute a function $f$. In particular, we will consider circuits with gates of fan-out and fan-in at most two: (i) AND, (ii) OR, (iii) NOT and (iv) duplication gate\footnote{This gate takes one bit as input and outputs two copies of the input bit.}. We will call such a circuit $(s,d)$-bounded if it has at most $s$ wires and has depth $d$. In this paper we almost exclusively deal with the case of $d=\tO(1)$. Also for uniformity, we will think of each input bit as a `constant gate.'

\section{The case of $k=2$}
\label{sec:k=2}

	In this section we consider the special case of $\cK=\{a,b\}$ (for Alice and Bob) but still over an arbitrary graph $G$. Our main result is Theorem~\ref{thm:2pSD}, which we prove in this section and is re-stated below:

\newtheorem*{thm:k2}{Theorem~\ref{thm:2pSD}}
\begin{thm:k2}[Restated]
        For any function $f:\{0,1\}^n\times \{0,1\}^n\to \{0,1\}$, and any graph $G$ we have that
        \[\tauroute(G,\{a,b\},\CC_{\eps}(f)) \le 4 {R_{\eps}(f,G,\{a,b\}}).\]
\end{thm:k2}
We would again like to stress that our proof does not proceed by invoking two party communication complexity lower bounds on two party functions induced by cuts on $G$. However, we do prove the result via a two-party communication simulation (where the two parties are denoted by $a'$ and $b'$). We will argue that if the result is not true, then we can come up with a two-party protocol for $f$ with cost strictly less than $\CC_{\eps}(f)$, which would lead to a contradiction.  Roughly speaking, if we can not route $\CC_\eps(f)/2$ units flow from $a$ to $b$ in time $2\tau := 2R_{\eps}(f, G, a, b)$, then we can remove a few special edges in $G$ to make the distance between $a$ and $b$ to be at least $2\tau$ (see Lemma~\ref{lemma:hop}).  Then we can divide the vertices in $G$ into $2\tau + 1$ levels indexed from $0$ to $2\tau$, such that $a$ is at level $0$, $b$ is at level $2\tau$ and all non-special edges are between two adjacent levels or between vertices in the same level. Thus, the flow of information via non-special edges is slow; in particular, without the special edges, the information from $a$ and $b$ will not mix in $\tau$ rounds.  Informally speaking, the important bits in the protocol are those sent via special edges. Thus, in the simulation using a two-party protocol between $a'$ and $b'$, we only send these important bits. By a careful analysis, we can bound the number of bits sent between $a'$ and $b'$ by less than $\CC_\eps(f)$, leading to a contradiction.

%

\begin{proof}[Proof of Theorem~\ref{thm:2pSD}]
	It would be convenient to consider the set $\vec{E}$ of directed edges, obtained from $E$ by replacing each edge $(u, v) \in E$ with two directed edges $(u,v)$ and $(v, u)$.  Consider the protocol in graph $G$ that computes the function $f$ in $\tau:=R_{\eps}(f,G,\{a,b\})$ steps.  For every $(u, v) \in \vec{E}$, let $x^t_{u, v}$ be the bit sent from $u$ to $v$ at time $t$ (recall that we allow both directions of an edge in $E$ to be used simultaneously).  The bit $x^t_{u,v}$ is a function of the bits received by $u$ by time $t-1$ (and the public random string); here we assume that $a$ received the input string $\vx_a$ and $b$ received the input string $\vx_b$ at time $0$. We assume towards the contradiction that $\tauroute\left(G,\{a,b\},\CC_{\eps}(f)\right) > 4\tau$.  By Claim~\ref{claim:tau-sub-additive}, we have $\tauroute\left(G,\{a,b\},N\right) > 2\tau$, where $N=\left \lceil \CC_{\eps}(f)/2 \right \rceil$.  This says that one cannot route $N$ bits in $2\tau$ rounds from $a$ to $b$. 
	
	\begin{lemma}
		\label{lemma:hop}
		Given a graph $G=(V, E)$ and $a, b \in V$,  assume there is no protocol that sends $N$ bits from $a$ to $b$ in $T$ rounds. Then there exists a vector $\ell \in  \{0, 1, 2, \cdots, T+1\}^V$ such that $\ell_a  = 0, \ell_b = T+1$ and
		\begin{align*}
		\sum_{(u, v) \in E} \max\{|\ell_u - \ell_v|-1, 0\}  < N.
		\end{align*}
	\end{lemma}
\begin{proof}
	We consider the time graph $\timeG{T}$. Since there is no protocol that sends $N$ bits from $a$ to $b$ in time $T$, we can not send $N$ units of flow from $(a, 0)$ to $(b, T)$ in $\timeG{T}$ (with congestion 1). By the max-flow-min-cut theorem, there is a cut of size strictly smaller than $N$ in $\timeG{T}$ that separates $(a, 0)$ from $(b, T)$. Let $(A, B)$ be the cut in $\timeG{T}$. For every $t \in \{0, 1, 2, \cdots, T\}$, let $A_t = \{v \in V: (v, t) \in A \}$.  Since there are infinitely many memory edges $((v, t),(v, t+1))$, no such edge can be cut, and we have that $a \in A_0 \subseteq A_1 \subseteq A_2 \subseteq \cdots \subseteq A_T \not\ni b$. 
	Now, the $(A, B)$ cut value is exactly 
	\begin{align*}
	\sum_{t = 0}^{T-1}\sum_{u \in A_t}\sum_{v \not\in A_{t+1}}\mathbf 1_{(u, v) \in E} < N.
	\end{align*}
	
	For each $t \in {1, 2, 3, \cdots, T}$, we define $V_t = A_t \setminus A_{t-1}$. Define $V_0 = A_0$ and $V_{T+1} = V \setminus A_{T}$.  Thus, $a \in V_0$ and $b \in V_{T+1}$ and $(V_0, V_1, V_2, \cdots, V_{T+1})$ forms a partition of $V$.   For each $v \in V$, let $\ell_v$ be the index such that $v \in V_{\ell_v}$.  We claim that in the above sum, each $(u,v) \in E$ is counted exactly $\max\{0, |\ell_u-\ell_v|-1\}$ times.  Without loss of generality assume $\ell_u \leq \ell_v$. Then, $u \in A_t$ and $v \not\in A_{t+1}$ for $\ell_u \leq t< \ell_v - 1$. Thus, $(u, v)$ is counted exactly $\max\{|\ell_u - \ell_v| - 1, 0\}$ times. Thus, we have 
	\begin{align*}
	\sum_{(u, v) \in E} \max\{|\ell_u - \ell_v|-1, 0\}  < N,
	\end{align*}
	which concludes our assertion.
\end{proof}
		
	Applying Lemma~\ref{lemma:hop} with $T=2\tau$, we obtain vector $\ell \in \{0, 1, 2, \cdots, 2\tau+1\}^V$ satisfying the properties stated in the lemma. We shall use a two-party protocol to simulate the protocol on $G$; we use $a'$ and $b'$ to denote the two parties participating in the two-party protocol. We assume $a'$ knows the input string $\vx_a$ and $b'$ knows the input string $\vx_b$. The two-party protocol has $\tau$ rounds that correspond to the $\tau$ time steps of the protocol on $G$ and is defined as follows.
	In each round $t$ for $t=1$ to $\tau$, for each bit $x^t_{u,v}$ sent from $u$ to $v$ in the original protocol on $G$,  $x^t_{u,v}$ is either sent from $a'$ to $b'$, or from $b'$ to $a'$, or not sent at all according to the following conditions:
	\begin{itemize}[itemsep=0pt,topsep=3pt]
		\item If $\ell_u < t < \ell_v$ then bit $x^t_{u,v}$ is sent from $a'$ to $b'$ in the two-party protocol.
		\item If $\ell_v < 2\tau + 1-t <  \ell_u$ then bit $x^t_{u,v}$ is sent from $b'$ to $a'$ in the two-party protocol.
		\item Otherwise bit $x^t_{u,v}$ is not sent in the two-party protocol.
	\end{itemize}
	We now claim that (i) the number of bits sent in the two-party protocol is at most $2N-2$ (where recall  $N=\left \lceil \CC_{\eps}(f)/2 \right \rceil$), and (ii)  the protocol is valid in the sense that the two parties can compute the bits transmitted during its execution, and once completed, both parties $a'$ and $b'$ know the output of the graph protocol.
	The fact that the round complexity is bounded by $2N-2$ follows directly by our definitions. Namely, in the two-party protocol, for any edge $e = (u, v) \in E$,  there are at most $2(|\ell_u-\ell_v|-1)$ different $t$'s for which the bit $x^t_{u,v}$ or $x^t_{v, u}$ is sent between $a'$ and $b'$. 
	As by Lemma~\ref{lemma:hop}, $\sum_{(u, v) \in E} \max\{|\ell_u - \ell_v|-1, 0\}  \le N-1$ we conclude (i) above.
	We now prove the validity of the protocol.
	
	\begin{lemma}
		\label{lemma:simulation}
		Let $t \in [0, \tau]$. (1) At the end of round $t$: if $\ell_v \leq 2\tau - t$, then $a'$ knows all the bits received by $v$; if $\ell_v \geq t+1$, then $b'$ knows all the bits received by $v$.  (2) If $t < \tau$ then in round $t+1$, $a'$ knows all the bits she needs to send to $b'$, and $b'$ knows all the bits he needs to send to $a'$.
	\end{lemma}
	
	\begin{proof}
		We first show that for each $t \in [0, \tau-1]$, (1) implies (2). If $a'$ needs to send $x^{t+1}_{u,v}$ to $b'$ in round $t+1$, then we must have $\ell_u < t+1$.  $x^{t+1}_{u,v}$ depends on all the bits received by $u$ by the end of round $t$. Since $\ell_u < t+1 < 2\tau - t$, (1) implies that $a'$ knows all these bits and thus can compute $x^{t+1}_{u,v}$. Thus, $a'$ knows all the bits she needs to send to $b'$ in round $t+1$; similarly, $b'$ knows all the bits he needs to send to $a'$. 
		
		We now prove the lemma by induction on $t$; for each $t$ we only need to prove (1). The base case is $t=0$; at the end of round 0, $a'$ knows all the bits received by $v$ if $v \neq b$ and $b'$ knows all the bits received by $v$ if $v \neq a$. So, (1) holds since $\ell_a = 0$ and $\ell_b = 2\tau + 1$. 
		
		 Consider some $t \geq 1$ and assume (1) holds for $t - 1$.  We prove (1) for $t$; we only need to prove the statement for $b'$, since the statement for $a'$ can be proved symmetrically.  Let $\ell_v \geq t+1$ and we need to prove that $b'$ knows all the bits received by $v$ before the end of round $t$.  Since $\ell_v \geq (t-1) + 1$, by the induction hypothesis, $b'$ knows all the bits received by $v$ before the end of round $t-1$. We only need to show that $b'$ knows all the bits received by $v$ at round $t$. 
		
		Focus on a vertex $u$ such that $(u, v) \in \vec{E}$. In the graph protocol, the bit $x^t_{u,v}$ is sent from $u$ to $v$ at time $t$. We consider two cases.  First consider the case that $\ell_u \geq t$. Thus $\ell_u \geq (t-1)+1$; by the induction hypothesis, $b'$ knows all the bits received by $u$ before the end of round $t-1$;  thus $b'$ can compute $x^t_{u, v}$.  For the other case, we have $\ell_u < t < \ell_v$. 
		In the two-party protocol, $a'$ sends $x^t_{u,v}$ to $b'$, implying that $b'$ knows $x^t_{u,v}$ by the end of round $t$ (notice that by induction hypothesis, (2) holds for $t-1$; thus $a'$ knows $x^t_{u,v}$). This finishes the proof of the lemma.
	\end{proof}	
	
	Lemma~\ref{lemma:simulation} implies point (ii) above.
	Indeed, at the end of round $\tau$, $a$ knows the output; $a'$ knows all bits received by $a$ as $\ell_a = 0 \leq 2\tau - \tau$. So $a'$ knows the output. Similarly $b'$ knows the output.  Notice that $2N-2 < \CC_{\eps}(f)$. The error of the two-party protocol on every input $(\vx_a, \vx_b)$ is exactly the same as the error of the graph protocol on this input.  Thus, we obtain a two-party protocol with total communication less than $\CC_{\eps}(f)$ and error $\epsilon$; this contradicts the definition of $\CC_{\eps}(f)$. So the theorem holds.
\end{proof}


\section{Steiner Tree Packing Bounds}
\label{sec:sd}


In this section we consider general sets $\cK$.
We first present a lower bound on $R_\eps(f, G, \cK)$ based on the notion of Steiner tree packing.
We then explore the potential optimality of conceptually simple protocols that perform computation of $f$ over a (collection of) Steiner trees that span $\cK$. 



We prove the following general lower bound result:
       \begin{thm}
                \label{thm:sd--gen-lb-main}
                Let $G, \cK, f:\big(\{0,1\}^n\big)^\cK \to \{0, 1\}$, $\eps\ge 0$ be defined as usual.
                Assume for some $n' > 0$ 
                that the following is true: for every pair of distinct players $a, b  \in \cK$, there exists some $\tilde \vx \in \big(\{0, 1\}^n\big)^{\cK \setminus \{a,b\}}$,  such that 
                $\CC_\eps(f_{\{a\}, \{b\},\tilde \vx}) \geq n'$. Then,
        \[                \min_{\Delta\in [|V|]} \left(\frac{n'}{\ST(G,\cK,\Delta)}+\Delta\right) \leq \tilde O\Big(R_\eps(f, G, \cK)\Big).\]
        \end{thm}
We provide an overview of the proof of the above result for the case of $\SD_{\cK,n}$. We first use Theorem~\ref{thm:2pSD} to get many edge disjoint paths between every pair of terminals in $\cK$ with length at most $\tO(\tau)$, where the optimal protocol takes $\tau$ rounds. (This follows from the fact that $\tauroute(G,\{a,b\},\CC_\eps(f'))\le 4\tau$ via Theorem~\ref{thm:2pSD}.) Then using tools developed in earlier work on packing Steiner trees with bounded diameter by Marathe et al.~\cite{MRS98}, we show that we can stitch these sets of edge disjoint path to obtain a large enough set of edge disjoint Steiner tree packings with diameter $\tO(\tau)$. This is enough to prove our lower bound for $\SD_{\cK,n}$. Next, we prove the result for general $f$.

\begin{proof}[Proof of Theorem~\ref{thm:sd--gen-lb-main}]
Assume that $R_\eps(f, G, \cK)=\tau$.
In particular, for any $a$ and $b$ in $\cK$ it holds that
$R_\eps(f_{\{a\},\{b\},\tilde \vx},G,\{a,b\}) \leq \tau$.
By Theorem~\ref{thm:2pSD}, for $f'=f_{\{a\},\{b\},\tilde \vx}$, and the fact that $n' \leq \CC_\eps(f')$,
\[\tauroute(G,\{a,b\},n') \leq \tauroute(G,\{a,b\},\CC_\eps(f')) \le 4 R_{\eps}(f',G,\{a,b\}) \leq 4\tau.\] 
Thus, there exists 
$n'$ edge-disjoint paths connecting $(a,0)$ and $(b,\const\tau)$ in $G^{(\const\tau)}$. This in turn implies $n'$ fractional edge disjoint paths in $G$ of length at most $\const\tau$ in which each path has fractional value $\frac{1}{\const\tau}$, yielding a total value of $\frac{n'}{\const\tau}$ fractional edge disjoint paths (of length at most $\const\tau$).
The analysis above holds for all pairs $a$ and $b$ in $\cK$.
In what follows (in Theorem~\ref{thm:matching} given below), we show that the latter implies a fractional Steiner Tree packing in $G$ of value $\tOmega\left(\frac{n'}{\tau}\right)$ with tree diameter at most $\tO(\tau)$. 
Implying that:
$$
\left( \min_{\Delta\in [|V|]} \left(\frac{n'}{\ST(G,\cK,\Delta)}+\Delta\right)\right) \leq \tO(\tau) = \tO(R_\eps(f, G, \cK)).
$$
\end{proof}

We now address the missing assertion in the proof of Theorem~\ref{thm:sd--gen-lb-main}.
We start with some notation. Given a (partial) matching $M$ over $\cK$ and a set $\cP$ of $|M|$ edge-disjoint paths in $G$, we say $\cP$ supports $M$ if for every $(a, b) \in M$, there is a path in $\cP$ connecting $a$ and $b$. We prove that 
\begin{thm}
\label{thm:matching}
Let $\cK=\{u_0,u_1,\dots,u_{k-1}\}$.
Assume that for every $u_i \in \cK \setminus \{u_0\}$ there is a collection $\cQ_i$ of fractional edge-disjoint paths of length at most $D$ from $u_i$ to $u_0$ in $G$ with total value $p$.  Then, there is a Steiner tree packing of value $\tOmega(p)$ in $G$ with tree diameter at most $\tO(D)$.
\end{thm}

\begin{proof}
We use the following lemma:

\begin{lemma}
\label{lemma:matching}
There is a randomized algorithm that given $\cK' \subseteq \cK$ of even cardinality outputs a matching $M$ over $\cK'$ and a set $ \cP$ of $|M|$ edge disjoint paths supporting $M$ such that (i) $|M| \geq |\cK'|/4$, (ii) all paths in $ \cP$ have length at most $16D$, and (iii) for every $e \in E$, $\Pr[e \text{ is used by paths in } \cP] \leq 4/p$.
\end{lemma}
	

\begin{proof} 
Let $E'$ be the set of all edges used by paths in $\cup_{u_i \in \cK'} \cQ_i$, let $w_e \leq 1$ be the total weight of paths in $\cQ_i$ that use $e$, and let $w(E')$ be the sum of edge weights of edges in $E'$.
So $w(E') \leq |\cK'|pD$.   
Let $G' = (V, E')$ with edge capacities $w_e$.
By \cite{li2004network} we can find a fractional Steiner tree packing $(\cT', z')$ of value $p/2$ in $G'$.  
However, there is no guarantee for the diameters of the trees in $ \cT'$.
Focus on each tree $T \in  \cT'$. 
It is not hard to find a perfect matching $M$ over $\cK'$, and a set of $|M|$ edge-disjoint paths $\cP$ in $T$ that supports $M$.  
Here, one needs to pair the elements of $\cK'$ iteratively starting from the pair with the least common ancestor which is furthest from a predefined root, removing that pair, and recursing.
We say a path $P \in \cP$ is short if its length is at most $16D$; otherwise, we say $P$ is long.  We say that $T$ is bad if the number of long paths in $\cP$ is at least $|\cK'|/4$; otherwise, we say $T$ is good. It follows that $\sum_{T \in  \cT': T \text{ bad}} z'_T \leq p/4$, as otherwise we have $w(E') > 16D \times |\cK'|/4 \times p/4 = |\cK'|pD$. A contradiction.
Thus, $Z'= \sum_{T \in  \cT': T \text{ good}} z'_T \geq p/4$. 

The randomized algorithm now works as follows. We first randomly choose a good tree $T\in  \cT'$ with probability $z'_T/Z'$.  
Then we take the perfect matching $M$ over $\cK'$ and the set of $|M|$ edge-disjoint paths $\cP$ in $T$ that support $M$. 
We remove all long paths from $\cP$ and their corresponding pairs from $M$. Then we output $(M, P)$. As each edge $e \in E'$ has $w_e \leq 1$, we have that 
$$
\Pr[e \text{ is used by paths in } \cP] = \frac{\sum_{T \in  \cT': T \text{good}, T \ni e} z'_T}{Z'} \leq  \frac{w_e}{Z'} \leq \frac{4}{p}.
$$
This finishes the proof of Lemma~\ref{lemma:matching}.
\end{proof}
	
We now proceed to the proof of Theorem~\ref{thm:matching}.  
We shall define a randomized algorithm to output a Steiner tree $T$ over $\cK$ of diameter at most $\tO(D)$. 
The final packing is implicitly defined by the randomized algorithm.  That is, a tree $T$ has $z_T$ value proportional to the probability that the randomized algorithm outputs $T$. The algorithm is a simple application of Lemma~\ref{lemma:matching} above and proceeds as follows: Initially, set $\cK' \gets \cK$, $T \gets \emptyset$. Now, repeat the following steps until $|\cK'| = 1$: (i) apply Lemma~~\ref{lemma:matching} to find a matching $M$ over $\cK'$ and its corresponding supporting paths $\cP$, (ii) add the edges in $\cP$ to $T$, and (iii) for every $(u, v) \in M$, arbitrarily remove one of the two vertices in $\{u, v\}$ from $\cK'$.  Finally, return $T$. Note that this procedure recurses $\left \lceil \log_{4/3}{k}\right\rceil \le 4\log{k}$  many times (and the final diameter and congestion in the worst-case gets multiplied by $4\log{k}$).
Lemma~\ref{lemma:matching} implies that the diameter of $T$ is at most $64 D \log k=\tO(D)$. Moreover, for every $e \in E$, the probability that $e \in T$ is at most $\frac{16}{p}\log k$. 
	
To obtain the fractional Steiner Tree packing, let $p_T$ be the probability that tree $T$ is returned by the randomized algorithm. It follows that $z_T=\frac{p}{16\log{k}}p_T$ is a solution to the Steiner Tree packing LP of value $\frac{p}{16\log{k}}=\tOmega(p)$. This finishes the proof of Theorem~\ref{thm:matching}.
\end{proof}

\subsection{Steiner tree upper bounds}
\label{app:sd-ub}

We consider a reasonably large class of composed functions. In particular, given a function $g:\{0,1\}^n\to \{0,1\}$, the class of functions $g\circ\SYM$ if the class of all functions $f:\left(\{0,1\}^n\right)^{\cK}\to\{0,1\}$ such that there exits `inner' symmetric functions $h_i:\{0,1\}^{\cK}\to\{0,1\}$ for $i\in [n]$ such that
\[f(\{\vx_u\}_{u\in \cK}) = g\left( h_1(\{\vx_u[1]\})_{u\in\cK},\dots, h_n(\{\vx_u[n]\})_{u\in\cK}\right).\]
Note that $\SD_{\cK,n}$ is a special case when $g$ is the $n$-bit $OR$ and $h_i$ is the $k$-bits $\AND$. Next, we argue that all such functions have a simple Steiner tree type upper bound.
We now show that
\begin{lemma}
\label{lemma:k-SD}
For any graph $G$ and subset of players $\cK$, let $f$ be in $g\circ \SYM$ for an arbitrary $g:\{0,1\}^n\to\{0,1\}$. Then
\[R_0(f,G,\cK)  \leq \tO\left( \min_{\Delta\in [|V|]} \left(\frac{n}{\ST(G,\cK,\Delta)}+\Delta\right)\right).\]
\end{lemma}

\begin{proof}
Let $\Delta \in [|V|]$ and consider an optimal fractional solution to the $\Delta$-diameter Steiner Tree (ST) packing LP of value  $\ST(G,\cK,\Delta)$.
Such a solution can be rounded to an integral ST packing of value $\tOmega(\ST(G,\cK,\Delta)$ \cite{raghavan1987randomized}.
Let $u_0 \in \cK$.
For every tree in the ST packing it is straightforward to schedule the transmission of a stream of bits from each terminal in $\cK \setminus \{u_0\}$ towards $u_0$ such that vertex $u_0$ receives for all $i\in [n]$, the sum of the $i$th bits in $\cK \setminus \{u_0\}$. Note that since the $h_i$'s are symmetric functions this is enough for $u_0$ to compute the value of $f$. If each terminal $u$ holds $m$ bits $\vx_u\in\{0,1\}^m$, using a single tree, vertex $u_0$ will be able to compute the sum of the collection $\{\vx_u\in\{0,1\}^m\}_{u\in\cK}$ in at most $m\left\lceil\log{k}\right\rceil+\Delta$ rounds.
Using the $\tOmega(\ST(G,\cK,\Delta)$ trees in parallel one may set $m=\tO\left(\frac{n}{\ST(G,\cK,\Delta)}\right)$ on each tree to conclude our assertion.
\end{proof}

\subsection{Some tight bounds}
\label{app:tight-ST}

As noted earlier, $\SD_{\cK,n}$ is a special case of the composed function from Section~\ref{app:sd-ub}. Lemma~\ref{lemma:k-SD} along with Theorem~\ref{thm:sd--gen-lb-main} (where we use the well-known lower bounds for two-party $\SD$~\cite{Razborov92} and setting $\tilde \vx$ to be the all $1$s vector) 
proves Theorem~\ref{thm:k-SD}. 

We sketch how this result can be extended to a larger family of composed functions.

\begin{prop}
\label{prop:gen-comp}
Consider the class of all composed functions (in the sense of Section~\ref{app:sd-ub}) where all the inner symmetric functions $h_i$'s are not the constant function, the parity (or its negation). Further, the outer function is such that $g(\vx[1]\vee\vy[1],\dots,\vx[n]\vee\vy[n])$ has two party-communication complexity of $\Omega(n)$. Then for any $\eps\ge 0$, every function $f$ in this class satisfies:
\[R_{\eps}(f,G,\cK)  = \tT\left( \min_{\Delta\in [|V|]} \left(\frac{n}{\ST(G,\cK,\Delta)}+\Delta\right)\right).\]
\end{prop}
Note that $\neg\SD_{\cK,n}$ belongs to this class of functions. 
\begin{proof}[Proof Sketch of Proposition~\ref{prop:gen-comp}]
Lemma~\ref{lemma:k-SD} proves an upper bound of $\tO\left( \min_{\Delta\in [|V|]} \left(\frac{n}{\ST(G,\cK,\Delta)}+\Delta\right)\right)$. Further, since the $h_i$'s are not one of the four ruled out functions, there is always a way to fix any $k-2$ of the inputs (other than say the terminals $a$ and $b$) such that value of $f$ is determined by $g'(\vx_a,\vx_b)=g(\vx_a[1]\vee\vx_b[1],\dots,\vx_a[n]\vee\vx_b[n])$. Indeed, by the choice of $h_i$, for every $i\in [n]$, there exist  a value $0\le c_i<k-1$ such that $h_i$ evaluates to different values on inputs with $c_i$ and $c_{i}+1$ ones. Further, it evaluates to the same value on inputs of size $c_{i}+1$ and $c_{i}+2$. In other words, if we pick $\tilde \vx$ such that the sum of the number of ones among $\tilde\vx[u]$ for all $u\in \cK\setminus \{a,b\}$ in the $i$th position is exactly $c_i$, then we note that $f_{\{a\},\{b\},\tilde \vx}$ is exactly $g(\vx_a[1]\vee\vx_b[1],\dots,\vx_a[n]\vee\vx_b[n])$.
By assumption $g'$ has  $\Omega(n)$ two party communication complexity, which by Theorem~\ref{thm:sd--gen-lb-main} implies an overall lower bound of $\tOm\left( \min_{\Delta\in [|V|]} \left(\frac{n}{\ST(G,\cK,\Delta)}+\Delta\right)\right)$.
\end{proof}

\section{Multicommodity flow type bounds}
\label{sec:ed}


\subsection{Circuits to Protocols}


Here we sketch the proof of Lemma~\ref{lem:ckt-protocol}, which we re-state below:
\newtheorem*{lem:ckt}{Lemma~\ref{lem:ckt-protocol}}
\begin{lem:ckt}[Restated]
Let $f:\left(\{0,1\}^n\right)^k\to \{0,1\}$ have a  circuit with constant fan-in and constant fan-out gates and depth $d$. Further, each level $i\in [d]$ has $s_i$ gates in it. Then
\begin{equation}
\label{eq:ckt-protocol-ub}
R_0(f,G,\cK) \le \sum_{i=1}^d \tO\left(\tau_{\MCF}\left(G,\cK,\frac{s_i}{k}\right)\right).
\end{equation}
Finally, we can upper bound the above by $\tO\left(d\cdot \tau_{\MCF}\left(G,\cK,\frac{s}{k}\right)\right)$ as well as $\tO\left(\frac{s}{k}\cdot \tau_{\MCF}\left(G,\cK,1\right)\right)$.
\end{lem:ckt}

The proof is an adaptation of an idea that was used in~\cite{DKO13} to design protocols for $G$ being a clique (i.e. the CONGEST-CLIQUE model).
Let $C$ be the given circuit for $f$. Then one can assign each gate of $C$ to each terminal in $\cK$ and then we evaluate each layer by setting up a multi-commodity flow problem where for each gate $g$ in the current level all the input gates (or their assigned terminals) send their value to $g$ (or the player that is assigned to $g$). Since at level $i$ $C$ has $s_i$ gates,  it can be shown via the probabilistic method that there exists an assignment of gates such that each terminal only has a total requirement of $\tO(s_i/k)$. 
We now present the details.

\begin{proof}[Proof of  Lemma~\ref{lem:ckt-protocol}]
 Assuming~\eqref{eq:ckt-protocol-ub} is correct, we note that the second bound follows by the simple observation that $s_i\le s$. Further, the third bound follows from Claim~\ref{claim:tau-sub-additive} and the fact that $\sum_{i=0}^d s_i\le s$.

We now argue~\eqref{eq:ckt-protocol-ub}. Let $C$ be the given $(s,d)$-bounded circuit for $f$. For every $0\le i\le d$, let $s_i$ be the number of gates a level $i$. (Note that $s_0=nk$ and $s_d=1$.) The idea is to evaluate the circuit $C$ in the given delay. We will do so by evaluating all gates in a given level one at a time. In particular, we will argue that we can evaluate the gates at level $i$ with delay $\tO\left(\tau_{\MCF}\left(G,\cK,\frac{s_i}{k}\right)\right)$. Note that this suffices to prove~\eqref{eq:ckt-protocol-ub}.

We will need a map from the gates of $C$ to terminals in $\cK$ with certain properties. To show the existence of such a map,
let $\pi$ denote a random map from the $s$ gates of $C$ to the $k$ players. We note that by a standard balls and bins argument, any set of $\Theta(s_i)$ gates are assigned to any specific player with load $L_i=O\left(\frac{s_i}{k}\log{kd}\right)$ with probability $>1-1/(2d)$. (We will see shortly that this is enough to handle all bad cases that may arise in the rest of our arguments.) We begin with level $0$. Note that in this case the $s_0=nk$ input bits would need to be re-routed according to $\pi$. By the balls and bins argument, this means we have a demand set where each player has load $L_0+s_0/k$. (Recall that initially each player has $n=s_0/k$ bits.) Thus, we can `evaluate' level $0$ with delay $\tau_{\MCF}\left(G,\cK,L_0+s_0/k\right)$, which by Claim~\ref{claim:tau-sub-additive} is $\tO\left(\tau_{\MCF}\left(G,\cK,\frac{s_0}{k}\right)\right)$, as desired.

Assume by induction we have evaluated all levels up to level $i\ge 0$. Now consider level $i+1$. Consider an arbitrary gate $g$ whose inputs are gates $g'$ (and possibly) $g''$. We add a demand pair with requirement $1$ between the pairs $(\pi(g),\pi(g'))$ and $(\pi(g),\pi(g''))$. Note that since there are $s_{i+1}$ such gates $g$ and at most $2s_{i+1}$ input gates from previous levels.
Thus, by the balls and bins argument, each player has at most $3L_{i+1}$ of the gates at level $i+1$ and their inputs. This implies that $\tau_{\MCF}\left(G,\cK,3L_{i+1}\right)$ rounds suffice to evaluate level $i+1$, which by Claim~\ref{claim:tau-sub-additive} is $\tO\left(\tau_{\MCF}\left(G,\cK,\frac{s_{i+1}}{k}\right)\right)$, as desired.

Finally we note that we had at most $2d-1$ bad events (where a bad event is at level $i$ some player has more than $L_i$ gates from level $i$ or one of its input gates assigned to it) that we would like $\pi$ to avoid.  By the union bound, there exists a map $\pi$ that makes the protocol above go through with the required round complexity.
\end{proof}


\subsection{The lower bound}



We are now ready to state our most general lower bound.
\begin{thm}
	\label{thm:ed-lb-main}
	Let $G, \cK, f:\big(\{0,1\}^n\big)^\cK \to \{0, 1\}, \eps\ge 0$ be defined as usual, and assume $k$ is even. Let $h: [k/2] \times [k/2] \to \R_{\geq 0}$. Assume the following is true: for every pair of disjoint sets $\cA, \cB  \subseteq \cK$ such that $|\cA|, |\cB| \leq k/2$, there exists some $\tilde \vx \in \big(\{0, 1\}^n\big)^{\cK \setminus (\cA \cup \cB)}$,  such that 
	$\CC_\eps\left(f_{\cA, \cB,\tilde \vx}\right) \geq h\big(|\cA|, |\cB|\big)$. Then,
	%
	\begin{align}
		\tau_{\MCF}\big(G, \cK, n'\big) \leq \tO\big(R_\eps(f, G, \cK)\big),
	\end{align}
	where $n' = \min_{y, z \in [0, k/2]: y + z > k/2 }\frac{h(y,z)}{y + z - k/2}$.
\end{thm}

The above implies lower bounds for the $\ED_{\cK,n}$ function: 

\begin{cor}
	\label{lem:ED-lb}
	For any $G$ and $\cK$, if $n \geq 1 + 2\ceil{\log k}$, then we have
	\begin{equation*}
		\tau_{\MCF}\big(G,\cK, 1\big)  \leq \tO\big(R(\ED_{\cK,n},G,\cK)\big),
		\qquad \text{and} \qquad
		\tau_{\MCF}\big(G,\cK, n\big)  \leq \tO\big(R_0(\ED_{\cK,n},G,\cK)\big).
	\end{equation*}
\end{cor}
\begin{proof}
        Let $f = \ED_{\cK,n}$. Fix some $\cA, \cB \subseteq \cK$ such that $\cA \cap \cB = \emptyset$ and $|\cA| \leq |\cB| \leq k/2$. We shall let $\tilde \vx \in \big(\{0, 1\}^{n}\big)^{\cK \setminus (\cA \cup \cB)}$ be a vector so that $\tilde \vx_{v, 1} = 1$ for every $v \in \cK \setminus (\cA \cup \cB)$, and  the $|\cK \setminus (\cA \cup \cB)|$ vectors  $\set{\tilde \vx_{v} }_{v \in \cK \setminus (\cA \cup \cB)}$ are different. This is possible since $n \geq 1 + 2\ceil{\log k}$. Then for the function $f_{\tilde \vx, \cA, \cB}(\vx_A, \vx_B)$, we are interested in the input pairs $(\vx_A, \vx_B)$ such that $\vx_{A,v}[1] = 0$ for every $v \in \cA$ and  $\vx_{B, v}[1] = 0$ for every $v \in \cB$. Thus, $f_{\tilde x, \cA, \cB}(x_A, x_B) = 1$ if and only if the $|\cA| + |\cB|$ strings $\{\vx_A[v]\}_{v \in \cA} \cup \{\vx_B[v]\}_{v \in \cB}$  are all  different.  In other words, we want to compute the two party $\SD$ problem on the sets $\{\vx_A[v]\}_{v\in \cA}$ and $\{\vx_B[v]\}_{v\in\cB}$. It is well-known that $\CC\left(f_{\cA, \cB,\tilde \vx}\right) \ge \Omega(|\cA|)$ (\cite{HW07}). We argue from first principles in Theorem~\ref{thm:EQ-k-by-k'} that $\CC_0\left(f_{\cA, \cB,\tilde \vx}\right) \ge \Omega(n|\cA|)$.

        Let $\eps=1/3$. Let $n' = \tilde \Omega(1)$ 
        be small enough. Let $h(y, z) = n' \min\{y, z\}$ for every $y, z \in [k/2]$.  Then $\min_{y, z \in [k/2]: y+z>k/2}\frac{h(y,z)}{y+z-k/2} = \min_{0 \leq y \leq z \leq k/2:  y+z>k/2}\frac{yn'}{y+z-k/2} = n'$. Thus, if $n'$ 
        is small enough, then the condition for Theorem~\ref{thm:ed-lb-main} holds. Thus, we have 
        \[\tau_{\MCF}\big(G, \cK, n'\big) \leq \tO\big(R(\ED_{\cK,n}, G, \cK)\big).\] 
        Then by Claim~\ref{claim:tau-sub-additive}, $\tau_{\MCF}(G, \cK, 1) \leq \ceil{\frac{1}{n'}}\tau_{\MCF}(G, \cK, n') \leq \tilde O(1) \tau_{\MCF}(G, \cK, n') \leq \tO\big(R(\ED_{\cK,n}, G, \cK)\big)$.

        Let $\eps=0$. Let $n' = \Omega(n)$ 
        be small enough. Let $h(y, z) = n' \min\{y, z\}$ for every $y, z \in [k/2]$.  Again, if $n'$ is small enough, then the condition for Theorem~\ref{thm:ed-lb-main} holds. Thus, we have 
\[\tau_{\MCF}\big(G, \cK, n'\big) \leq \tO\big(R_0(\ED_\cK, G, \cK, n)\big).\] 
Then $\tau_{\MCF}(G, \cK, n) \leq \ceil{\frac{n}{n'}}\tau_{\MCF}(G, \cK, n') \leq  O(1) \tau_{\MCF}(G, \cK, n') \leq \tO\big(R_0(\ED_{\cK,n}, G, \cK)\big)$, by Claim~\ref{claim:tau-sub-additive}.
\end{proof}

In Section~\ref{app:apps} we make use of other corollaries of Theorem~\ref{thm:ed-lb-main}.

We now sketch the proof of Theorem~\ref{thm:ed-lb-main} (specialized to $R(\ED_{\cK,n},G,\cK)$). 
First we note that any cut separating $k'$ terminals from the rest of the $k-k'$ terminals induces a two party communication complexity problem that needs $\Omega(\min(k',k-k'))$ communication across the cut. This in conjunction with our argument for $k=2$ implies that there are $\Omega(\min(k',k-k'))$ edge disjoint paths between the two subsets in $G^{(\tau)}$. We now use the cut-matching game framework of Khandekar, Rao and Vazirani~\cite{KRV09} to argue that we can construct an expander graph each of whose edges can be embedded into $G^{(\tau)}$ in the sense that each edge in the expander corresponds to a path in $G^{(\tau)}$ (and these paths have low congestion). Since the multicommodity flow with a total demand requirement of $\tO(k)$ from each terminal can be done with $d=\tO(1)$ delay on the expander graph, we can route these paths in $G^{(d\cdot \tau)}$. (We need to make sure that the paths in the expander are not too long but this can be done.) This implies a protocol for the multi-commodity flow problem that we need to solve for the upper bound with delay $\tO(\tau)$, as desired.

We now formally prove Theorem~\ref{thm:ed-lb-main}.  Fix any two disjoint sets $\cA, \cB \subseteq \cK$ such that $|\cA|, |\cB| \leq k/2$. Let $\tilde \vx \in \big(\{0,1\}^n\big)^{\cK \setminus (\cA \cup \cB)}$ be the vector satisfying the condition of the theorem for the pair $(\cA, \cB)$.  

By Theorem~\ref{thm:2pSD}, we have that 
	\[\tauroute(G_{A,B},\{v_A,v_B\},\CC_{\eps}(f_{A,B, \tilde \vx})) \le 4 {R_{\eps}(f_{A,B,\tilde \vx},G,\{v_A,v_B\}}).\]
where we overload notation for $R_{\eps}$ and $\CC_{\eps}$ by allowing input function to have inputs from different domains: i.e. unlike the original definition, $f_{\cA, \cB,\tilde \vx}:\left(\{0,1\}^n\right)^{\cA}\times \left(\{0,1\}^n\right)^{\cB}$ has two inputs from different domains.\footnote{This is the only place in this paper where we will need this overloading of notation.}

It is easy to see that  $R_\eps\big(f_{\cA, \cB, \tilde \vx}, G_{\cA, \cB}, \big\{v_\cA, v_\cB\big\}\big) \leq R_\eps\big(f, G, \cK\big)$, since every protocol to compute $f$ in $G$ among $\cK$ leads to a protocol to compute $f_{A, B, \tilde \vx}$ in $G_{\cA,\cB}$ between $v_\cA$ and $v_\cB$. Let $\tau =  \ceil{4R_\eps\big(f, G, \cK\big)}$. Then,  $\tauroute\Big(G_{\cA, \cB}, \{v_\cA, v_\cB\}, \CC_\eps\big(f_{\cA, \cB, \tilde \vx}\big)\Big)\leq \tau$.

 Since $\CC_\eps\big(f_{\cA, \cB, \tilde \vx}\big) \geq h\big(|\cA|,|\cB|\big)$, we have $\tauroute\Big(G_{\cA, \cB}, v_\cA, v_\cB, h\big(|\cA|,|\cB|\big)\Big) \allowbreak \leq  \tau$. Thus,  there are $h\big(|\cA|,|\cB|\big)$ edge disjoint paths from $(v_\cA, 0)$ to $(v_\cB, \tau)$ in $G_{\cA, \cB}^\tau$.   This implies that there are $h\big(|\cA|,|\cB|\big)$ edge-disjoint paths from $\cA \times \{0\}$ to $\cB \times \{\tau\}$ in $G^\tau$. To see this, focus on each of the $h\big(|\cA|,|\cB|\big)$ edge-disjoint paths from $(v_\cA, 0)$ to $(v_\cB, \tau)$ in $G_{\cA, \cB}^\tau$. Let $t$ be the smallest number such that $(v_\cB, t)$ is in the path; let $t'$ be the largest number such that $t' < t$ and $(v_\cA, t')$ is in the path.  Then, we modify this path as follows: we travel from $(v_\cA, 0)$ to $(v_\cA, t')$ using memory edges and then then use the segment of the path from $(v_\cA, t')$ to $(v_\cB, t)$, and then travel from $(v_\cB, t)$ to $(v_\cB, \tau)$ using the memory edges. After the modifications, the $h\big(|\cA|,|\cB|\big)$ edge-disjoint paths in $G_{\cA, \cB}^\tau$ can be naturally mapped back to $h\big(|\cA|,|\cB|\big)$ edge-disjoint paths in $G^\tau$ from $\cA \times \{0\}$ to $\cB \times \{\tau\}$. Next, we argue that these paths have even more structure.

\begin{lemma}
	\label{lemma:edps-between-AB}
	For partition $(\cA, \cB)$ of $\cK$ such that $|\cA| = |\cB| = k/2$,  we can find $n'k/2$ edge-disjoint paths from $\cA \times \{0\}$ to $\cB \times \{\tau\}$ in $G^\tau$, such that every vertex in $\cA \times \{0\}$ is the origin of exactly $n'$ paths, and every vertex in $\cB\times \{\tau\}$ is the destination of exactly $n'$ paths. 
\end{lemma}

\begin{proof}
	Construct a directed graph $\tilde G$ as follows. We start from $G^\tau$,  and add a super source $s$ and a super sink $t$. Then for every $u \in \cA$, we add $n'$ edges from $s$ to $(u, 0)$. For every $v \in \cB$, we add $n'$ edges from $(v, \tau)$ to $t$. To prove the lemma, it suffices to show that there are $n'k/2$ edge-disjoint paths from $s$ to $t$ in $\tilde G$.  Assume otherwise. Then, there is an $s$-$t$ cut $(S, T)$ in $\tilde G$ whose size is strictly less than $n'k/2$. 	Let $\cA'$ be the subset of $\cA$ such that $S \cap (\cA \times \{0\}) = \cA' \times \{0\}$; let $\cB'$ be the subset of $\cB$ such that $T \cap (\cB \times \{\tau\}) = \cB' \times \{\tau\}$. The number of edges  in the cut  that are incident to $s$ or $t$ is exactly $n' \Big(\big|T \cap (\cA \times \{0\})\big| + \big|S \cap (\cB \times \{\tau\})\big|\Big) = n'(|\cA \setminus \cA'| + |\cB \setminus \cB'|) = n'(k - |\cA'| - |\cB'|)$.  It implies that the number of edges in $G^\tau$ in the $(S, T)$ cut is strictly less than $n'k/2 -  n'(k - |\cA'| - |\cB'|) = n'\big(|\cA'| + |\cB'|-k/2\big) \leq h(|\cA'|,|\cB'|)$, by the definition of $n'$. (Note that if $|\cA'|+|\cB'|\le k/2$ then the inequality is trivially true since $h$ is always positive.)
Thus, we find a cut in the original graph $G^\tau$ of size strictly less than $h(|\cA'|,|\cB'|)$ separating $\cA' \times \{0\}$ and $\cB' \times \{\tau\}$, a contradiction.  This finishes the proof of the lemma.
\end{proof}

	We use the cut-matching game of Khandekar, Rao and Vazirani~\cite{KRV09}. In this game, we are given a set $V_X$ of $N_X$ vertices, where $N_X$ is even, and two players: a cut player, whose goal is to construct an expander $X = (V_X, E_X)$ on the set $V_X$ of vertices, and a matching player, whose goal is to delay its construction. The game is played in iterations. We start with the graph $X = (V_X, \emptyset)$.
	
	In each iteration $j$, the cut player computes a bi-partition $(\cA_j,\cB_j)$ of $V_X$ into two equal-sized sets, and the matching player returns some perfect matching $M_j$ between the two sets. The edges of $M_j$ are then added to $E_X$. Khandekar, Rao and Vazirani have shown that there is a strategy for the cut player, guaranteeing that after $O(\log^2N_X)$ iterations we obtain a $1/2$-expander with high probability.  Subsequently, Orecchia et al.~\cite{OSV08} have shown the following improved bound:
	
	\begin{thm}[Cut-Matching Game \cite{OSV08}]\label{thm:CMG}
		There is a probabilistic algorithm for the cut player, such that, no matter how the matching player plays, after $O(\log^2N_X)$ iterations, graph $X$ is an $\alphaCMG(N_X) = \Omega(\log N_X)$-expander, with constant probability.
	\end{thm}
	
	\begin{defn}
		Let $E'$ be a set of edges over $\cK$, $\tilde \tau > 0$ be an integer. Let $\overline{E'}$ be the set of directed edges obtained from $E'$ by replacing every undirected edge $e = (u, v) \in E'$ with two directed edges $(u, v)$ and $(v, u)$. An \emph{embedding} of $E'$ in $G^{\tilde \tau}$ is a set $\cP = \set{P_e: e \in \overline{E'}}$ of paths, where $P_e$ for a directed edge $e = (u, v)$ is a path connecting $(u, 0)$ to $(v, \tilde \tau)$ in $G^{\tilde \tau}$.
	\end{defn}
	
	\begin{lemma}\label{lemma:construct-expander}
		There is a randomized algorithm that outputs an $O(\log^2k)$-regular $\Omega(\log k)$-expander $X = (\cK, E_X)$, and an embedding $\cP$ of $E_X$ in $G^\tau$, such that 
		the expected number of paths in $\cP$ that use each edge $e$ in $G^\tau$ is at most $O(\log^2k/n')$. 
	\end{lemma}
	
	\begin{proof}
We run the cut-matching game over $\cK$. Initially, $\cP = \emptyset$ and $E_X = \emptyset$. 
		
		In the $j$-th iteration of the game, the cut-player finds a partition $(\cA_j, \cB_j)$ of $\cK$ according to his strategy. Then by Lemma~\ref{lemma:edps-between-AB}, we can find a set $\cQ_j$ of $n'|\cA_j| =n' k/2$ edge-disjoint paths from $\cA_j \times \{0\}$ to $\cB_j \times \{\tau\}$ in $G^\tau$, such that every vertex in $\cA_j \times \{0\}$ is the origin of exactly $n'$ paths and every vertex in $\cB_j \times \{\tau\}$ is the destination of exactly $n'$ paths.  These paths naturally define an $n'$-regular bipartite graph $H = (\cA_j, \cup \cB_j, E_H)$ between $\cA_j$ and $\cB_j$, where for each edge $e = (u, v) \in E_H, u \in \cA_j, v \in \cB_j$, $e$ is associated with a unique path $Q_e \in \cQ_j$ connecting $(u, 0)$ to $(v, \tau)$ in $G^\tau$. We can break $E_H$ into $n'$ matchings between $\cA_j$ and $\cB_j$.  Then, the matching player will randomly choose a matching $M_j$, out of the $n'$ matchings, each with probability $1/n'$. The matching player will play $M_j$; so we shall add $M_j$ to $E_X$. 
		
		Let $\cQ'_j  = \set{Q_e: e \in M_j}$ be the set of paths corresponding to $M_j$, and let $\cQ''_j$ be the set of mirrored paths of paths in $\cQ'_j$. The mirrored edge of an edge $((u, t-1), (v, t))$ in $G^\tau$ is the edge $((v, \tau-t), (u, \tau-t+1))$. The mirrored path of a path $P$ is constructed by concatenating the mirrored edges of all edges in $P$.  Thus, if $P$ connects $(u, 0)$ to $(v, \tau)$ in $G^\tau$, then the mirrored edge of $P$ connects $(v, 0)$ to $(u, \tau)$ in $G^\tau$. Thus, $\cQ'_j \cup \cQ''_j$ is an embedding of $M_j$ in $G^\tau$.  Since paths in $\cQ_j$ are edge-disjoint, each edge in $G^\tau$ belongs to $\cQ'_j$ with probability at most $1/n'$. Thus, each edge belongs to $\cQ''_j$ with probability at most $1/n'$. Moreover, $\cQ'_j \cup \cQ''_j$ causes congestion at most $2$ in $G^\tau$.   We add $\cQ'_j \cup \cQ''_j$ to $\cP$.
		
		Considering all the $O(\log^2 k)$ iterations together, $\cP$ is an embedding of $E_X$ in $G^\tau$. The paths in $\cP$ cause congestion at most $O(\log^2 k)$, and the expected number of paths in $\cP$ that use an edge $e$ in $G^\tau$ is at most $O(\log^2k)/n'$. By Theorem~\ref{thm:CMG}, the graph $X$ we obtained is an $O(\log^2 k)$-regular $\alphaCMG(k)$-expander. The algorithm succeeds with constant probability and thus we can repeat the algorithm until it succeeds.  The expected number of times we run the algorithm is a constant; this can only increase the expected number of paths in $\cP$ that use an edge by a constant factor. 
	\end{proof}
	
	We emphasize that we are not interested in the efficiency of the algorithm in Lemma~\ref{lemma:construct-expander} as it is only used for the analysis. Indeed, we need an exponential time algorithm to check whether $X$ is an $\alphaCMG(k)$-expander or not since the problem is NP-hard. 
	
	We use Lemma~\ref{lemma:construct-expander} to find a $d$-regular $\Omega(\log k)$-expander $X = (\cK, E_X)$, for some $d = O(\log^2 k)$, and an embedding $\cP = \set{P_e:e \in \overline{E_X}}$ of $E_X$ in $G^\tau$.  Let $A$ be the adjacency matrix of $X$ and $\lambda_2$ be the second largest eigenvalue of $A$.  Then, by Cheeger's Inequality, we have $\Phi(X) \leq \sqrt{2d(d-\lambda_2)}$.  Thus $\lambda_2 \leq d-\phi^2(X)/(2d) \leq d-\Omega(1)$, since $\phi^2(X)/(2d) = \Omega(\log^2 k)/O(\log^2 k) = \Omega(1)$.
	
	We consider the lazy random walk on $X$, starting from some vertex $v \in \cK$. By Theorem~\ref{thm:random-walk},  the difference between the distribution we obtain after $T$ steps of random walk and the uniform distribution is at most $\sqrt{k}\left(\frac{1+\lambda_2/d}{2}\right)^T$, in terms of the $L_1$ distance. Notice that $\frac{1+\lambda_2/d}{2} \leq \frac{2-\Omega(1/d)}{2} = 1 - \Omega\left(\frac{1}{\log^2k}\right)$. If we let $T = O(\log^3 k)$ to be large enough, then the difference is at most $1/(2k)$.  Thus, after $T$ steps of the lazy random walk, the probability that we are at each vertex $u \in \cK$ is at least $1/(2k)$. 
	
	Using the random walk, we show how to send $1/(2k)$ units of flow from $v$ to $u$ in $X$, for every ordered pair $(v, u) \in \cK^2$.  We have $k$ types of commodity, indexed by $\cK$. Initially, for every vertex $v \in \cK$, $v$ has 1 unit of commodity $v$. At each time step we do the following. For every $v \in \cK$,  and a commodity type $v' \in \cK$, we send $1/(2d)$ fraction of commodity $v'$ to each of the $d$ neighbors of $v$; thus, $1/2$ fraction of the commodity $v'$ will remain at $v$. After $T$ steps, every vertex $u \in \cK$ has at least $1/(2k)$ units of commodity $v'$, for every $v' \in \cK$.  Since $X$ is regular, at each time, the total amount of commodity at each vertex $v$ is 1. In each step, the amount of commodity sent through each edge $e \in E_X$ in each direction is exactly $1/(2d)$.
	
	Now, we can simulate the flow in the time graph $G^{\tau'}$, for $\tau' = T\tau$.  Recall that $\cP= \{P_E: e \in \overline{E_X}\}$ is the embedding of $E_X$ in $G^\tau$.  Initially, for each vertex $v \in \cK$ and a commodity type $v' \in \cK$, there is 1 unit of commodity $v'$ at $(v, 0)$.  Suppose at the $t$-th step, we sent $x$ units of commodity $v'$ from $v \in \cK$ to its neighbor $u \in \cK$, using edge $e$ in $E_X$. Let $e' \in \overline{E_X}$ be edge $e$ directed from $v$ to $u$. Then in graph $G^{\tau'}$, we sent $x$ units of commodity $v'$ from $(v, (t-1)\tau)$ to $(u, t\tau)$, using the path $P_{e'}$, shifted by $(t-1)\tau$ units of time. That is, the shifted path contains $((v', (t-1)\tau + i-1), (u', (t-1)\tau + i))$, for every edge $((v', i-1), (u', i))$ in $P_{e'}$. If $x$ units of commodity $v'$ remains at $v$, then we send $x$ units of commodity $v'$ from $(v, (t-1)\tau)$ to $(v, t\tau)$ using the memory edges at $v$.  Thus, we have a multi-commodity flow in $G^{\tau'}$, where for each ordered pair $(v, u) \in \cK^2$, we sent at least $1/2k$ units of flow from $(v, 0)$ to $(u, \tau')$. 
	
	If an edge $((v', i-1), (u', i)$ in $G^\tau$ is used by $p$ paths in $\cP$, then for every $t \in [T]$, the amount of flow sent through the $((v', (t-1)\tau + i - 1), (v', (t-1)\tau + i ))$ is $p/(2d)$.  By Lemma~\ref{lemma:construct-expander}, the expected amount of flow sent through each edge $e$ in $G^\tau$ is at most $1/(2d) \times O(\lg^2k/n') = O(1/n')$, where the expectation is over the randomness of $X$ and $\cP$.  Taking all pairs $(X, \cP)$ in the probability space (again, we are not interested in the efficiency of the algorithm), and scaling the multi-commodity flow by a factor of $2n'$, we obtain a multi-commodity flow in $G^{\tau'}$, where for each ordered pair $(v, u) \in \cK^2$, we sent at least $n'/k$ units of flow from $(v, 0)$ to $(u, \tau')$.  The flow causes congestion $O(1)$ in $G^{\tau'}$.   By, scaling $\tau'$ by a constant factor, we can reduce the congestion to $1$.  This proves that $\tau_{\MCF}(G, \cK, n') \leq O(T\tau) = O(\lg^3k)\cdot \left\lceil 4 R_{\eps}(f,G, \cK)\right\rceil = \tilde O(R_{\eps}(f, G, \cK))$, finishing the proof of Theorem~\ref{thm:ed-lb-main}. 

\subsection{Bounds for $\ED$}


%

The proof of the upper bound in Theorem~\ref{thm:ED} will crucially use the following result on existence of a small circuit for $\ED$: 
\begin{lemma}
\label{lem:ed-ckt-bound}
$\ED_{\cK,m}$ has an $(O(km\log{k}), O(m\log{k}))$-bounded circuit.
\end{lemma}
\begin{proof}
We first recall that there exists sorting networks that sort $k$ numbers with $O(k\log{k})$ swaps and depth $O(\log{k})$~\cite{AKS83}. By swap we mean a gate that takes as input two numbers and outputs the smaller number as the ``first" output and the larger number as the ``second" output. Note that if the numbers are $m$-bits then such a swap can be implemented with an $(O(m),O(m))$ bounded circuit. This implies that there exists a $(O(km\log{k}), O(m\log{k}))$-bounded circuit to sort $k$ numbers (each of which is $m$ bits).

Assume that the sorted numbers are $\vx_1,\dots,\vx_k$. Then note that the final answer is
\[\wedge_{i=1}^{k-1} \neg\EQ(\vx_i,\vx_{i+1}),\]
where $\EQ(\vx,\vy)=1$ if and only if $\vx=\vy$. Note that one can implement the $\EQ$ function with a $(O(m),O(\log{m}))$-bounded circuit. This implies that we can compute $\ED_{\cK,m}(\vx_1,\dots,\vx_k)$ with a $(O(km),O(\log{km}))$ bounded circuit (assuming $\vx_1,\dots,\vx_k$ are sorted in that order).

Thus, combining the two circuits, we get an $(O(km\log{k}), O(m\log{k}))$-bounded circuit for $\ED_{\cK,m}$, as desired.
\end{proof}

It is known that $\ED_{\cK, n}$ can be solved by solving $\ED_{\cK,O(\log{k})}$ (by using $O(\log{k})$ random hashes for each input)-- see e.g.~\cite{CRR14}. By Lemma~\ref{lem:ed-ckt-bound}, there exists a randomized $(O(k\log^2{k}),O(\log^2{k}))$-bounded circuit to solve $\ED_{\cK,n}$. Lemma~\ref{lem:ckt-protocol} and Claim~\ref{claim:tau-sub-additive} then show that $R(\ED_{\cK,n},G,\cK) \le \tO\left(\tau_{\MCF}(G,\cK,1)\right)$. Similarly using Lemma~\ref{lem:ed-ckt-bound} with $m=n$ we have that  $R_0(\ED_{\cK,n},G,\cK) \le \tO\left(\tau_{\MCF}(G,\cK,n)\right)$. Note that these upper bounds match the lower bounds in Corollary~\ref{lem:ED-lb}, which in turn proves Theorem~\ref{thm:ED}.


\section{Applications}
\label{app:apps}



We now consider distributed graph problems. For such problems every player $u\in \cK$ receives a subgraph $H_u$ and the goal of the players is to compute some (Boolean) function on the overall graph
\[H\eqdef \bigcup_{u\in \cK} H_u.\]

We define $N_H$, $M_H$ and $\Delta_H$ to be the number of vertices in $H$, number of edges in $H$ and the maximum degree in $H$ respectively. We will present our bounds in terms of these parameters (as well as parameters that depend on the underlying topology).

\subsection{Distribution of the input}
\label{sec:repr}

In this section, we tackle issues related to how the inputs $\{H_u\}_{u\in\cK}$ are represented and distributed among the players in $\cK$. We will assume that $H_u$'s (and hence $H$) are presented in the adjacency list representation and that all players know the set of vertices $V(H)$. In other words, the only knowledge that is distributed is the set of edges $E(H)$. There are two natural ways of distributing  the edges set that we consider in this section:
\begin{enumerate}
\item \textit{Node distribution:} In this case the adjacency list of a vertex is assigned to a terminal in $\cK$ as a whole. Further, we will assume that for every $u\in V(H)$, all terminals know the location of the assigned terminal for $u$.\footnote{This is a relatively mild assumption since these mappings in practical applications are done by publicly known hash mappings.} However, only the assigned terminal knows the adjacency list of $u$.
\item \textit{Edge distribution:} In this case the edge set $E(H)$ is distributed among the $k$ terminals and in this case all the terminals only know about the identity of $V(H)$.
\end{enumerate}

Finally, we will assume that in either distribution all of the $H_u$'s are roughly of the same size.
\begin{defn}
A node (edge resp.) distribution of $H$ among the $k$ players is called $M$-balanced if for every $u\in \cK$, the size of $H_u$ is at most $M$.\footnote{In the case of node distribution, the size of $H_u$ is the sum of the degree of the vertices assigned to $u$ while in the case of edge distribution, the size of $H_u$ is the number of edges assigned to $u$.}
\end{defn}

It turns out that one can convert a balanced edge distribution into a random balanced node distribution.
\begin{lemma}
\label{lem:edge->node-distr}
If $H$ is represented by an $\tO(M_H/k+\Delta_H)$-balanced edge distribution then it can be converted into an $\tO(M_H/k+\Delta)$-balanced node distribution in $\tO(\tau_{\MCF}(G,\cK,M_H/k+\Delta_H))$ rounds of communication. Further, in the latter, every node is assigned uniformly and independently at random to the terminals in $\cK$.
\end{lemma}
\begin{proof}
The argument basically follows from a technical result in~\cite{klauck}. Let $\pi:V(H)\to \cK$ be a completely random map (i.e. each vertex is mapped independently and uniformly randomly to $\cK$). Then~\cite[Lemma 4.1]{klauck} argues that size of the newly mapped $H_u$ is $\tO(M_H/k+\Delta_H)$. It is easy to see that we can move from the edge distribution to the random node distribution with a multicommodity flow problem with $\tO(M_H/k+\Delta_H)$-bounded demands, which completes the proof.
\end{proof}

It turns out that the extra pre-processing round complexity of $\tO(\tau_{\MCF}(G,\cK,M_H/k+\Delta_H))$ can always be absorbed in the upper bounds that we can prove and so for the rest of the section, when talking about upper bounds we will assume that $H$ is node distributed such that each node is randomly assigned a terminal in $\cK$. Note that this implies that our upper bounds hold for worst-case balanced node or edge distribution. However, our upper bounds do not hold when the distribution of $H$ over the terminals is {\em skewed}. Skew is a known issue in parallel processing and handling it is left as an open problem.

Our lower bounds work for both $\tO(M_H/k+\Delta_H)$-balanced node and edge distribution representations. However, unlike the results of~\cite{klauck}, our lower bounds assume a worst-case partition of the input among the terminals.

\subsection{Some hard problems}

In this section, we define some hard problems that we will reduce to our distributed graph problems.

The two problems, which we dub $\ORSD_{\cK, n}$ and $\ANDSD_{\cK, n}$ respectively, informally are the logical $\OR$ (and logical $\AND$ resp.) of $\binom{k}{2}$ independent copies of the two-party $\SD$ problem. In particular, each player $u\in \cK$ gets $k-1$ strings $\{\vx_{u,v}\}_{v\in\cK\setminus \{u\}}$. Then the players want to compute
\[\ORSD_{\cK, n}\left(\{ \vx_{u,v}\}_{u\in \cK, v\in \cK\setminus \{u\}}\right)=\bigvee_{\{u,v\}\in \binom{\cK}{2}} \left( \bigvee_{i\in [n]} \vx_{u,v}[i] \wedge \vx_{v,u}[i]\right),\]
and
\[\ANDSD_{\cK, n}\left(\{ \vx_{u,v}\}_{u\in \cK, v\in \cK\setminus \{u\}}\right)=\bigwedge_{\{u,v\}\in \binom{\cK}{2}} \left( \bigvee_{i\in [n]} \vx_{u,v}[i] \wedge \vx_{v,u}[i]\right),\]
where for a set $S$, we use $\binom{S}{2}$ to denote the set of all unordered pairs from $S$.

We show the hardness of the two above functions by recalling the large communication complexity of two closely related functions in the classical two-party model: Let Alice (Bob) get $m$ strings, $\vx_1,\ldots,\vx_m$ ($\vy_1,\ldots,\vy_m$), with each $\vx_i \in \{0,1\}^n$ ($\vy_i \in \{0,1\}^n$). Let $\ORSDTWO_{m, n}$ denote the problem of determining if any pair of strings $(\vx_i,\vy_i)$ have a 1 at a common index. Then, 
the following is a simple implication of Bar-Yossef et.al \cite{BYJKS04}.

\begin{thm}  \label{thm:OR-DISJ-2-party}
 $\CC_{1/3}\big(\ORSDTWO_{m, n} \big) \ge \Omega(mn)$.
\end{thm}

Similarly, define $\ANDSDTWO_{m,n}$ as the 2-party problem of determining if all pairs of strings $(\vx_i,\vy_i)$ have a 1 at a common index. This is also called the $\TRB_{m,n}$ problem. The following establishes its hardness.

\begin{thm}[Jayram et al.\cite{JKS03}]  \label{thm:Tribes-2-party}
 $R_{1/3}^{(2)}\big(\TRB_{m,n} \big) \ge \Omega(mn)$.
\end{thm}

Theorem~\ref{thm:ed-lb-main} implies the following results:
\begin{cor}
\label{cor:or-disj}
For any $G$ and $\cK$, we have
\[R(\ORSD_{\cK, n},G,\cK)\ge \tOm\big(\tau_{\MCF}(G,K,nk)\big).\]
\end{cor}
\begin{proof}
        Let $f = \ORSD_{\cK,n}$. Fix some $\cA, \cB \subseteq \cK$ such that $\cA \cap \cB = \emptyset$ and $|\cA|, |\cB| \leq k/2$. We shall let $\tilde \vx \in \big(\{0, 1\}^{nk}\big)^{\cK \setminus (\cA \cup \cB)}$ be an all-0 vector.  Note that $f_{\cA,\cB,\tilde\vx}$ is exactly an $\ORSDTWO_{|\cA|\cdot|\cB|,n}$ problem. Thus, by Theorem~\ref{thm:OR-DISJ-2-party}, we have that $\CC_{1/3}\left(f_{\cA, \cB,\tilde \vx}\right) \ge \Omega(|\cA|\cdot|\cB|\cdot n)$.

        Let $\eps=1/3$. Let $n'' = \Omega(n)$ be small enough; let $n' = n''k/2$; let $h(y, z) = n''yz$ for every $y, z \in [k/2]$.  Then $\min_{y, z \in [k/2]: y+z>k/2}\frac{h(y,z)}{y+z-k/2} = \min_{y,z \in [k/2]:  y+z>k/2}\frac{n''yz}{y+z-k/2} = n''k/2=n'$, where the second equality holds since $(k/2-y)(k/2-z) \geq 0$ implies $yz \geq (y + z - k/2)k/2$, and $y=1, z=k/2$ implies $\frac{yz}{y+z-k/2}=k/2$.

        Thus, if $n''$ is small enough, then the condition for Theorem~\ref{thm:ed-lb-main} holds. Thus, we have 
\[\tau_{\MCF}\big(G, \cK, n'\big) \leq \tO\big(R(\ORSD_{\cK,n}, G, \cK)\big).\] 
Then $\tau_{\MCF}(G, \cK, kn) \leq \ceil{\frac{kn}{n'}}\tau_{\MCF}(G, \cK, n') \leq O(1) \tau_{\MCF}(G, \cK, n') \leq \tO\big(R(\ORSD_{\cK,n}, G, \cK)\big)$, by Claim~\ref{claim:tau-sub-additive}.
\end{proof}

\begin{cor}
\label{cor:and-disj}
For any $G$ and $\cK$, we have
\[R(\ANDSD_{\cK, n},G,\cK)\ge \tOm\big(\tau_{\MCF}(G,K,nk)\big).\]
\end{cor}
\begin{proof}
        Let $f = \ANDSD_{\cK,n}$. Fix some $\cA, \cB \subseteq \cK$ such that $\cA \cap \cB = \emptyset$ and $|\cA|, |\cB| \leq k/2$. We shall let $\tilde \vx \in \big(\{0, 1\}^{nk}\big)^{\cK \setminus (\cA \cup \cB)}$ be an all-1 vector.  Note that $f_{\cA, \cB,\tilde \vx}$ is a $\TRB_{|\cA|\cdot|\cB|,n}$ problem. Thus, by Theorem~\ref{thm:Tribes-2-party}, we have that $\CC_{1/3}\left(f_{\cA, \cB,\tilde \vx}\right) \ge \Omega(|\cA|\cdot|\cB|\cdot n)$. The rest of the proof is the same as that of Corollary~\ref{cor:or-disj} and is omitted.
\end{proof}


\subsection{Reductions from $\ORSD$}

In this section we consider the following three problems:

\paragraph{Acyclicity.} Given $H_u$ to each player $u\in \cK$, the players have to decide if $H$ is acyclic or not.

\paragraph{Triangle-Detection.} Given $H_u$ to each player $u\in \cK$, the players have to decide if $H$ has a triangle or not.

\paragraph{Bipartiteness.} Given $H_u$ to each player $u\in \cK$, the players have to decide if $H$ is bipartite or not.

The argument below follows from a simple adaptation of the reduction used to prove hardness of these problems for the total communication case in~\cite{CRR14}.

\begin{thm}
\label{thm:or-disj-redux}
Each of the problems of acyclicity, triangle-detection and bipartiteness for input $H$ on topology $G$ with set of player $\cK$ needs $\tOm\left(\tau_{\MCF}\left(G,\cK,\frac{M_H+N_H}{k}\right)\right)$ rounds of communication (even for randomized protocols). Further, these results hold for the case when $H$ is $O(M_H/k+\Delta_H)$-balanced node (or edge) distributed.
\end{thm}
\begin{proof}
We will use a reduction from $\ORSD_{\cK, n}$ to construct an instance $H$ such that either (i) $H$ is a forest or (ii) $H$ has a triangle (depending on the output of the $\ORSD$ instance). Note that a protocol for any of acyclicity, triangle-detection or bipartiteness can distinguish between the two cases. Thus, to complete the proof we present the construction of $H$ from a given instance of $\{ \vx_{u,v}\}_{u\in \cK, v\in \cK\setminus \{u\}}$ of $\ORSD_{\cK,n}$. We will argue explicitly for node distribution and mention where the reduction needs to be modified to make it work for edge distribution.

Fix any $u\in \cK$.
We will define the subgraph $H_u$. $H_u$ is the disjoint union of subgraphs $H_{u,w}=(V_{u,w}, E_{u,w})$ for every $w\in\cK\setminus \{u\}$. In particular, $V_{u,w}$ consists of one vertex for each domain element of the universe corresponding to the two party $\SD$ corresponding to $(u,w)$ and two special vertices corresponding to the pair $\{u,w\}$. In other words we have
\[V_{u,w} = \left\{\cup_{i\in [n]} x^{\{u,w\}}_i\right\} \cup\{ y^{u,w}, y^{w,u}\}.\]
The edge set $E_{u,w}$ consists of the edge $( y^{u,w}, y^{w,u})$ plus edges between elements that are present in $\vx_{u,w}$ and $y^{u,w}$. In other words,
\[E_{u,w} = \left\{ \left(x^{\{u,w\}}_i, y^{u,w}\right)| \vx_{u,w}[i]=1\right\} \cup \left\{ ( y^{u,w}, y^{w,u})\right\}.\footnote{For edge distribution, we assign the edge $( y^{u,w}, y^{w,u})$ to exactly one of $H_u$ or $H_w$.}\]
See Figure~\ref{fig:or-disj-redux} for an illustration of this reduction.

To complete the argument we make the following observations. First if $\ORSD_{\cK, n}\left(\{ \vx_{u,v}\}_{u\in \cK, v\in \cK\setminus \{u\}}\right)=1$, then $H$ has a triangle otherwise $H$ is a forest. Indeed first note that for every $\{u,w\}\in\binom{\cK}{2}$, the subgraphs $H_{u,w}\cup H_{w,u}$ are node disjoint. Thus, $H$ has a triangle if and only if $H_{u,w}\cup H_{w,u}$ has a triangle for some $\{u,w\}\in\binom{\cK}{2}$. Next, we note that if $(\vx_{u,w}[i]\wedge \vx_{w,u}[i])=1$ for some $i\in [n]$, then the triple $\{y^{u,w},y^{w,u},x_i^{u,w}\}$ forms a triangle. Otherwise, $y^{u,w}$ and $y^{w,u}$ are connected via edges to disjoint set of the vertices in $\{x_i^{\{u,w\}}\}_{i\in [n]}$, which implies that $H_{u,w}\cup H_{w,u}$ is a forest. This argues the correctness of the reduction.

Second, for every $u\in \cK$, the player $u$ can construct $H_u$ from its input $\{ \vx_{u,w}\}_{w\in \cK\setminus \{u\}}$. Finally, note that in this construction both $N_H,M_H$ are $\Theta(nk^2)$.
Further, each $H_u$ is of size $O(nk)$, which is $O(M_H/k+\Delta_H)$.
All of the above along with Corollary~\ref{cor:or-disj} completes the proof.
\footnote{The fact that $M_H$ is $\Omega(nk^2)$ follows from the fact that sets in the hard distribution in Corollary~\ref{cor:or-disj} have sets whose size is linear in the size of the universe.}
\end{proof}

\begin{figure}[ht]
\begin{center}
\begin{tikzpicture}[scale=1.75]


\foreach \xS in {0,5.2}
{
	\foreach \yS in {0,5}
	{
		\draw (0.2+\xS,0.2+\yS) rectangle (1.4+\xS, 0.6+\yS); 
		\foreach \x in {1,2,3}
		{
			\draw (0.2+\xS+\x*0.4-0.4, 0.2+\yS) -- (0.2+\xS+\x*0.4-0.4, 0.6+\yS);
			\node [above] at (0.2+\xS+\x*0.4-0.2, 0.6+\yS)  {\textcolor{gray}{{\tiny $\x$}}};
		}
	}
}


\foreach \y in {1,2,3}
{
                \node [left] at (1.8, 1.2+\y*0.4-0.2) {\textcolor{gray}{{\tiny $\y$}}};
                \node [right] at (5.0, 1.2+\y*0.4-0.2) {\textcolor{gray}{{\tiny $\y$}}};
} 

\foreach \xS in {0, 2.8}
{
	\draw (1.8+\xS, 1.2) rectangle (2.2+\xS,2.4);
	\foreach \y in {1,2,3}
	{
		\draw (1.8+\xS, 1.2+\y*0.4-0.4) -- (2.2+\xS, 1.2+\y*0.4-0.4);
	}
}



\node at (0.4,0.4) {$0$};
\node at (0.8,0.4) {$1$};
\node at (1.2,0.4) {$0$};

\node [below] at (0.8,0.2) {$\vx_{2,1}$};


\node[circle,draw] (x112) at (0.4,1.2) {{\tiny $x_1^{\{1,2\}}$}};
\node[circle,draw] (x212) at (0.8,1.8) {{\tiny $x_2^{\{1,2\}}$}};
\node[circle,draw] (x312) at (1.2,1.2) {{\tiny $x_3^{\{1,2\}}$}};
\node[circle,draw] (y21) at (0.8,2.6) {{\tiny $y^{2,1}$}};

\draw (x212) edge (y21);




\node at (0.4,5.4) {$1$};
\node at (0.8,5.4) {$0$};
\node at (1.2,5.4) {$1$};

\node [above] at (0.8,5.8) {$\vx_{1,2}$};


\node[circle,draw] (x121) at (0.4,4.2) {{\tiny $x_1^{\{1,2\}}$}};
\node[circle,draw] (x221) at (0.8,4.8) {{\tiny $x_2^{\{1,2\}}$}};
\node[circle,draw] (x321) at (1.2,4.2) {{\tiny $x_3^{\{1,2\}}$}};
\node[circle,draw] (y12) at (0.8,3.4) {{\tiny $y^{1,2}$}};

\draw (y12) edge (y21);
\draw (y12) edge (x121);
\draw (y12) edge (x321);

\begin{scope}[shift={(5.2,0)}]



\node at (0.4,0.4) {$0$};
\node at (0.8,0.4) {$0$};
\node at (1.2,0.4) {$1$};

\node [below] at (0.8,0.2) {$\vx_{3,1}$};


\node[circle,draw] (x113) at (0.4,1.8) {{\tiny $x_1^{\{1,3\}}$}};
\node[circle,draw] (x213) at (0.8,1.2) {{\tiny $x_2^{\{1,3\}}$}};
\node[circle,draw] (x313) at (1.2,1.8) {{\tiny $x_3^{\{1,3\}}$}};
\node[circle,draw] (y31) at (0.8,2.6) {{\tiny $y^{3,1}$}};

\draw (x313) edge (y31);




\node at (0.4,5.4) {$1$};
\node at (0.8,5.4) {$1$};
\node at (1.2,5.4) {$0$};

\node [above] at (0.8,5.8) {$\vx_{1,3}$};


\node[circle,draw] (x131) at (0.4,4.2) {{\tiny $x_1^{\{1,3\}}$}};
\node[circle,draw] (x231) at (0.8,4.8) {{\tiny $x_2^{\{1,3\}}$}};
\node[circle,draw] (x331) at (1.2,4.2) {{\tiny $x_3^{\{1,3\}}$}};
\node[circle,draw] (y13) at (0.8,3.4) {{\tiny $y^{1,3}$}};

\draw (y31) edge (y13);
\draw (x131) edge (y13);
\draw (x231) edge (y13);

\end{scope}



\node at (2,1.4) {$0$};
\node at (2,1.8) {$1$};
\node at (2,2.2) {$1$};

\node [above] at (2,2.4) {$\vx_{2,3}$};


\node[circle,draw] (x123) at (2.6,0.6) {{\tiny $x_1^{\{2,3\}}$}};
\node[circle,draw,fill=orange, fill opacity=0.5] (x223) at (2.6,1.4) {{\tiny $x_2^{\{2,3\}}$}};
\node[circle,draw] (x323) at (2.6,2.6) {{\tiny $x_3^{\{2,3\}}$}};
\node[circle,draw,fill=orange, fill opacity=0.5] (y23) at (3.0,2) {{\tiny $y^{2,3}$}};

\draw (x223) edge (y23);
\draw (x323) edge (y23);




\node at (4.8,1.4) {$0$};
\node at (4.8,1.8) {$1$};
\node at (4.8,2.2) {$0$};

\node [above] at (4.8,2.4) {$\vx_{3,2}$};


\node[circle,draw] (x132) at (4.2,0.6) {{\tiny $x_1^{\{2,3\}}$}};
\node[circle,draw,fill=orange, fill opacity=0.5] (x232) at (4.2,1.4) {{\tiny $x_2^{\{2,3\}}$}};
\node[circle,draw] (x332) at (4.2,2.6) {{\tiny $x_3^{\{2,3\}}$}};
\node[circle,draw,fill=orange, fill opacity=0.5] (y32) at (3.8,2) {{\tiny $y^{3,2}$}};

\draw (x232) edge (y32);
\draw (y23) edge (y32);

			

\draw [dotted, gray, thick] (0.1, -0.1) rectangle (3.3, 2.95);
\draw [dotted, gray, thick] (3.5, -0.1) rectangle (6.7, 2.95);
\draw [dotted, gray, thick] (0.1, 3.05) rectangle (6.7, 6.2);

\node [above] at (3.4,3) {{\tiny Player $1$}};
\node [left] at (3.3, .1) {\textcolor{gray}{{\tiny Player $2$}}};
\node [right] at (3.5, .1) {\textcolor{blue}{{\tiny Player $3$}}};

\end{tikzpicture}
\end{center}
\caption{Illustration of the reduction in proof of Theorem~\ref{thm:or-disj-redux} for $n=k=3$. 
In this example the overall graph $H$ has a triangle and the three participating nodes are colored in orange. (Note that in this case $\ORSD_{\{1,2,3\},3}$ is $1$.)}
\label{fig:or-disj-redux}
\end{figure}
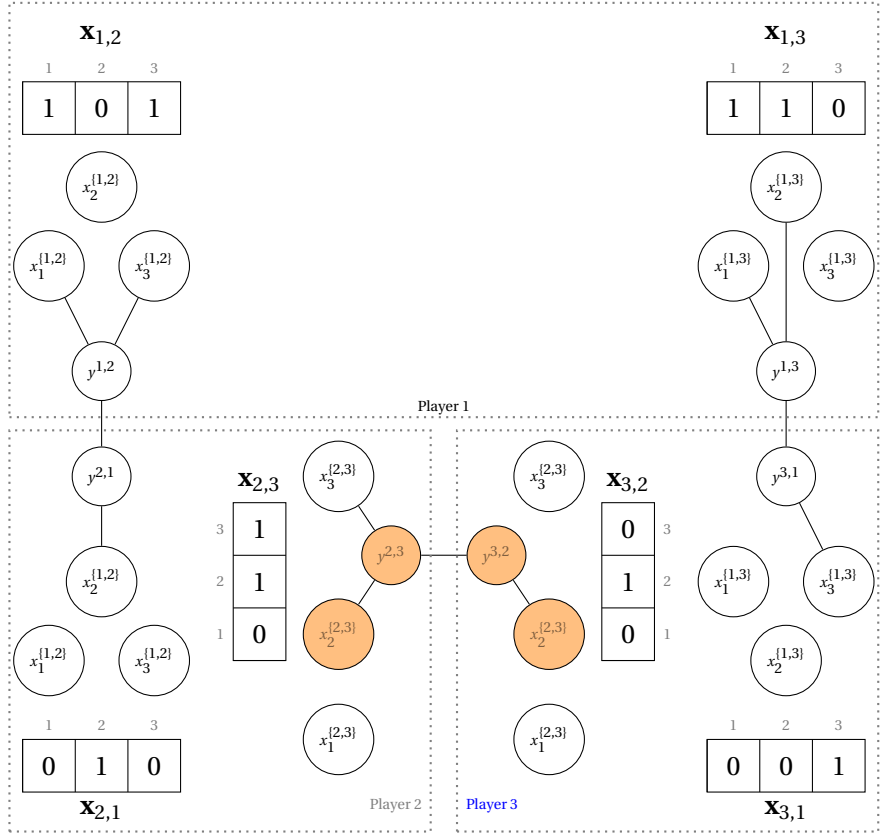

\subsubsection{Upper Bounds}

We defer the discussion of the upper bounds for acyclicity and bipartiteness to Section~\ref{app:bfs}.

We next outline a protocol (which is simple generalization of the protocol in~\cite{DKO13})  to detect whether $H$ contains a triangle or not.
\begin{prop}
\label{prop:triang}
Assuming that for every $\eps>0$, there exists arithmetic circuits of size $O(n^{2+\epsilon})$ for computing $n\times n$ matrix multiplication over $\mathbb{F}_2$, the problem of triangle detection on $H$ can be solved with $\tO\left(\tau_{\MCF}\left(G,\cK,\frac{(N_H)^{2+\epsilon}}{k}\right)\right)$ rounds of randomized communication. 
\end{prop}
Note that the above bound is within any polynomial factor of the lower bound in Theorem~\ref{thm:or-disj-redux} for the case of graphs with $M_H\ge \Omega(N_H^2)$.

\begin{proof}[Proof Sketch of Proposition~\ref{prop:triang}]
First, recall that cubing the adjacency matrix of $H$ over the Boolean semiring is enough to detect triangles. This is because a triangle is present if and only if the cubed matrix has a non-zero diagonal entry. It can be shown (see Section 2.1 of \cite{DKO13}) that there exists a randomized reduction of this problem to a few matrix multiplications over the field $\mathbb{F}_2$. Now the conjecture about matrix multiplication yields arithmetic circuits of $O(n^{2+\epsilon})$ size for these matrix multiplications. A further argument shows, exploiting the structure of matrix multiplication \cite{BCS97}, that such circuits can be made to have few wires and poly-logarithmic depth. Given such a circuit, an application of our Lemma~\ref{lem:ckt-protocol} yields the distributed protocol with $\tO\left(\tau_{\MCF}\left(G,\cK,\frac{(N_H)^{2+\epsilon}}{k}\right)\right)$ rounds. 
\end{proof}

\subsection{Reductions from $\ANDSD$}

In this section we consider the following two problems:

\paragraph{Connectivity.} Given $H_u$ to each player $u\in \cK$, the players have to decide if $H$ is connected or not.

\paragraph{Connected Components.} Given $H_u$ to each player $u\in \cK$, the players have to compute the number of connected components of $H$.

Since a lower bound for connectivity implies a lower bound for the connected components, we only present the lower bound for the latter. This reduction again is a simple adaptation of the corresponding one for total communication in~\cite{CRR14}.

\begin{thm}
\label{thm:and-disj-redux}
The connectivity problem for input $H$ on topology $G$ with set of player $\cK$ needs $\tOm\left(\tau_{\MCF}\left(G,\cK,\frac{M_H+N_H}{k}\right)\right)$ rounds of communication (even for randomized protocols).  Further, these results hold for the case when $H$ is $O(M_H/k+\Delta_H)$-balanced node (or edge) distributed.
\end{thm}
We note that by Proposition~\ref{prop:mcf-clique}, the above implies a lower bound of $\tOm((M_H+N_H)/k^2)$ for the case when $G$ is a clique on $k$ terminals. This quantitatively recovers the bound for connectivity proved in~\cite{klauck}.

\begin{proof}[Proof of Theorem~\ref{thm:and-disj-redux}]
We will use a reduction from $\ANDSD_{\cK, n}$ to construct an instance $H$ such that  $H$ is connected or not depending on output of the $\ANDSD_{\cK,n}$ instance. 
To complete the proof we present the construction of $H$ from a given instance of $\{ \vx_{u,v}\}_{u\in \cK, v\in \cK\setminus \{u\}}$ of $\ANDSD_{\cK,n}$. (The argument holds for both node and edge distributions.)

Fix any $u\in \cK$.
We will define the subgraph $H_u$. $H_u$ is the union of subgraphs $H_{u,w}=(V_{u,w}, E_{u,w})$ for every $w\in\cK\setminus \{u\}$. We next present the description of $H_{u,w}$. For the rest of the proof, we will assume that there is a pre-determined total order among the players, i.e. given any two $u,w\in\cK$, the comparison $u<w$ is well-defined.

In particular, $V_{u,w}$ consists of one vertex for each domain element of the the universe corresponding to the two party $\SD$ corresponding to $(u,w)$ and two special vertices corresponding to the  pair $\set{u,w}$. In other words, we have if $u<w$
\[V_{u,w} = \left\{\cup_{i\in [n]} x^{\{u,w\}}_i\right\} \cup\{\ell^{\{u,w\}},r\}\]
and otherwise
\[V_{u,w} = \left\{\cup_{i\in [n]} x^{\{u,w\}}_i\right\} \cup\{r\},\]
where the node $r$ is shared across all subgraphs.
The edge set $E_{u,w}$ consists of the following: Consider the case $u<w$. If $\vx_{u,w}[i]=1$, then the edge $\big(x^{\{u,w\}}_i,\ell^{\{u,w\}}\big)$ is present. If $\vx_{u,w}[i]=0$, the edge $\big(x^{\{u,w\}}_i,r\big)$ is present. In the other case of $u>w$, edge $\big(x^{\{u,w\}}_i,r\big)$ is present if $\vx_{u,w}[i]=1$. See Figure~\ref{fig:and-disj-redux} for an illustration of this reduction.

To complete the argument we make the following observations. First note that if $\ANDSD_{\cK, n}\left(\{ \vx_{u,v}\}_{u\in \cK, v\in \cK\setminus \{u\}}\right)=1$, then $H$ is connected otherwise $H$ is not. Second, for every $u\in \cK$, the player $u$ can construct $H_u$ from its input $\{ \vx_{u,w}\}_{w\in \cK\setminus \{u\}}$. Finally, note that in this construction both $N_H,M_H$ are $\Theta(nk^2)$. 
Further, each $H_u$ is of size $O(nk)$, which is $O(M_H/k+\Delta_H)$.
All of the above along with Corollary~\ref{cor:and-disj} completes the proof.\footnote{The claim that $M_H\ge \Omega(nk^2)$ follows from the fact that in the hard distribution in Corollary~\ref{cor:and-disj}, the individual sets are of size $\Omega(n)$.}
\end{proof}

\begin{figure}[ht]
\begin{center}
\begin{tikzpicture}[scale=1.75]



\draw (0.2,0.2) rectangle (1.4, 0.6);
\foreach \x in{1,2,3}
{
	\draw (0.2+\x*0.4-0.4, 0.2) -- (0.2+\x*0.4-0.4, 0.6);
        \node [above] at (0.2+\x*0.4-0.2, 0.6)  {\textcolor{gray}{{\tiny $\x$}}};
}

\node at (.4,.4) {$1$};
\node at (.8,.4) {$0$};
\node at (1.2,.4) {$0$};

\node [below] at (.8,.2) {$\vx_{2,1}$};



\node[circle, draw] (x121) at (.4,1.2) {{\tiny $x_1^{\{1,2\}}$}};
\node[circle, draw] (r2) at (.8,1.8) {{\tiny $r$}};

\draw (x121) edge (r2);



\draw (2, 1) rectangle (2.4,2.2);

\foreach \y in {1,2,3}
{
	\draw (2, 1+\y*.4-.4) -- (2.4, 1+\y*.4-.4);
	\node [left,gray] at (2, 1+\y*.4-.2) {{\tiny $\y$}};
}

\node at (2.2, 1.2) {$0$};
\node at (2.2, 1.6) {$1$};
\node at (2.2, 2) {$0$};

\node [below] at (2.2,1) {$\vx_{2,3}$};


\node[circle, draw] (x123) at (1.4,1.2) {{\tiny $x_1^{\{2,3\}}$}};
\node[circle, draw] (x323) at (1.4,2) {{\tiny $x_3^{\{2,3\}}$}};
\node[circle, draw] (x223) at (2.8,1.6) {{\tiny $x_2^{\{2,3\}}$}};
\node[circle, draw] (l23) at (2.8,.8) {{\tiny $\ell^{\{2,3\}}$}};

\draw (x123) edge (r2);
\draw (x323) edge (r2);
\draw (x223) edge (l23);





\draw (4, 1) rectangle (4.4,2.2);

\foreach \y in {1,2,3}
{
        \draw (4, 1+\y*.4-.4) -- (4.4, 1+\y*.4-.4);
        \node [left,gray] at (4, 1+\y*.4-.2) {{\tiny $\y$}};
}

\node at (4.2, 1.2) {$0$};
\node at (4.2, 1.6) {$1$};
\node at (4.2, 2) {$1$};

\node [below] at (4.2,1) {$\vx_{3,2}$};


\draw (5,0.2) rectangle (6.2, 0.6);

\foreach \x in{1,2,3}
{
        \draw (5+\x*0.4-0.4, 0.2) -- (5+\x*0.4-0.4, 0.6);
        \node [above] at (5+\x*0.4-0.2, 0.6)  {\textcolor{gray}{{\tiny $\x$}}};
}

\node at (5.2,.4) {$0$};
\node at (5.6,.4) {$0$};
\node at (6,.4) {$1$};

\node [below] at (5.6,.2) {$\vx_{3,1}$};


\node[circle,draw] (x232) at (5,1.2) {{\tiny $x_2^{\{2,3\}}$}};
\node[circle,draw] (x332) at (5,2) {{\tiny $x_3^{\{2,3\}}$}};
\node[circle,draw] (x331) at (5.8,1.2) {{\tiny $x_3^{\{1,3\}}$}};
\node[circle,draw] (r3) at (5.8,2) {{\tiny $r$}};

\draw (r3) edge (x232);
\draw (r3) edge (x332);
\draw (r3) edge (x331);





\draw (.2, 3) rectangle (.6,4.2);

\foreach \y in {1,2,3}
{
        \draw (.2, 3+\y*.4-.4) -- (.6, 3+\y*.4-.4);
        \node [left,gray] at (.2, 3+\y*.4-.2) {{\tiny $\y$}};
}

\node at (.4, 3.2) {$0$};
\node at (.4, 3.6) {$1$};
\node at (.4, 4) {$1$};

\node [below] at (.4,3) {$\vx_{1,2}$};


\draw (5.8, 3) rectangle (6.2,4.2);

\foreach \y in {1,2,3}
{
        \draw (5.8, 3+\y*.4-.4) -- (6.2, 3+\y*.4-.4);
        \node [right,gray] at (6.2, 3+\y*.4-.2) {{\tiny $\y$}};
}

\node at (6, 3.2) {$1$};
\node at (6, 3.6) {$1$};
\node at (6, 4) {$1$};

\node [below] at (6,3) {$\vx_{1,3}$};



\node[circle, draw] (x112) at (1.2,3) {{\tiny $x_1^{\{1,2\}}$}};
\node[circle, draw] (r3) at (2.6,3) {{\tiny $r$}};

\draw (x112) edge (r3);

\node[circle, draw, fill=blue, fill opacity=0.5] (x212) at (1.2,3.6) {{\tiny $x_2^{\{1,2\}}$}};
\node[circle, draw, fill=blue, fill opacity=0.5] (x312) at (1.2,4.2) {{\tiny $x_3^{\{1,2\}}$}};
\node[circle, draw, fill=blue, fill opacity=0.5] (l2) at (2.6,3.9) {{\tiny $\ell^{\{1,2\}}$}};

\draw (l2) edge (x212);
\draw (l2) edge (x312);



\node[circle, draw] (x113) at (5.2,3) {{\tiny $x_1^{\{1,3\}}$}};
\node[circle, draw] (x213) at (5.2,3.6) {{\tiny $x_2^{\{1,3\}}$}};
\node[circle, draw] (x313) at (5.2,4.2) {{\tiny $x_3^{\{1,3\}}$}};
\node[circle, draw] (l13) at (4,3.6) {{\tiny $\ell^{\{1,3\}}$}};

\draw (l13) edge (x113);
\draw (l13) edge (x213);
\draw (l13) edge (x313);
\draw (l13) edge (x113);




\draw [dotted, gray, thick] (0, -0.1) rectangle (3.3, 2.5);
\draw [dotted, gray, thick] (3.5, -0.1) rectangle (6.5, 2.5);
\draw [dotted, gray, thick] (0, 2.6) rectangle (6.5, 4.6);

\node [above] at (3.4,2.6) {{\tiny Player $1$}};
\node [left] at (3.3, .1) {\textcolor{gray}{{\tiny Player $2$}}};
\node [right] at (3.5, .1) {\textcolor{blue}{{\tiny Player $3$}}};

\end{tikzpicture}
\end{center}
\caption{Illustration of the reduction in proof of Theorem~\ref{thm:and-disj-redux} for $n=k=3$. (For clarity singleton nodes in each of the player's subgraphs are not shown.) 
In this example the overall graph $H$ has two connected components: the nodes shaded blue form one connected component and the rest of the vertices form another connected component. (Note that in this case $\ANDSD_{\{1,2,3\},3}$ is $0$.)}
\label{fig:and-disj-redux}
\end{figure}
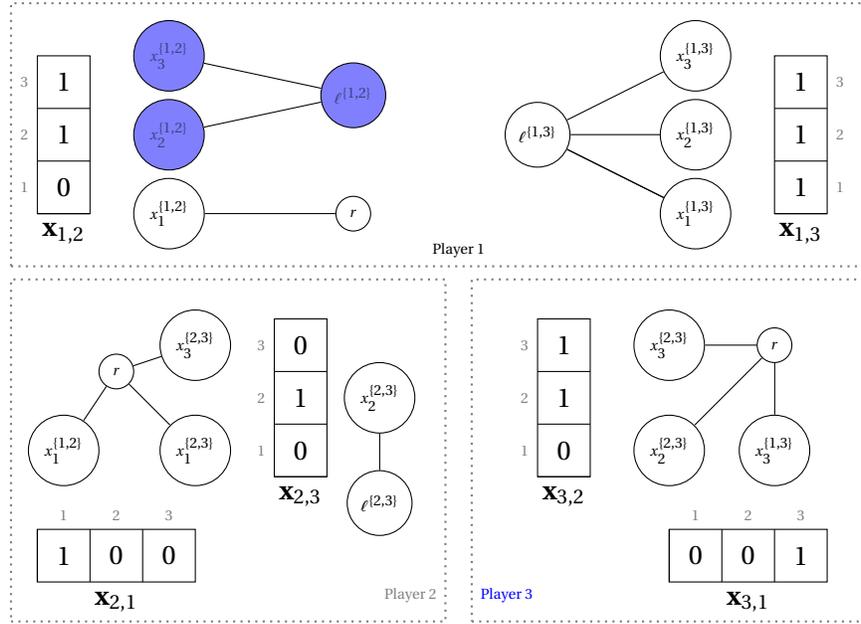

\subsubsection{Upper Bounds}
\label{app:bfs}

We outline how we can adapt the argument of~\cite{klauck} to implement BFS in our framework and then argue that for  large enough inputs $H$, the lower bound for connectivity in Theorem~\ref{thm:and-disj-redux} is tight. (Recall that we are assuming that the original input $H$ is randomly partitioned across the terminals in a node distribution: we'll call this the random node distribution.)

\begin{thm}
\label{thm:connectivity-ub} 
Let $H$ be a random node distributed graph. Then if $H$ is large enough compared to $G$, we can solve the connectivity problem on $H$ with $\tO\left(\tau_{\MCF}\left(G,\cK,\frac{M_H+ N_H}{k}\right)\right)$ randomized rounds of communication.
\end{thm}

Before we prove Theorem~\ref{thm:connectivity-ub}, we will need the following fact about multi commodity flows (which is only needed to prove the tightness of our upper bound):

\begin{lemma}
\label{lem:MCF-large-dem}
There exists a constant $c>0$ such that
given any $G$ and $\cK$, there exists an integer $B_0$ such that for every $B\ge B_0$ we have that
\[\tau_{\MCF}(G,\cK,B) \ge c\cdot\left(\frac{B}{B_0}\cdot \tau_{\MCF}(G,\cK,B_0)\right).\]
\end{lemma}

We remark that the above is not implied by Claim~\ref{claim:tau-sub-additive}. In particular, note that for $n''\le n'$, Claim~\ref{claim:tau-sub-additive} only shows that $\tau_{\MCF}(G,\cK,n'')\le \tau_{\MCF}(G,\cK,n')$, which is not enough to prove Lemma~\ref{lem:MCF-large-dem}.



\begin{proof}[Proof Sketch of Theorem~\ref{thm:connectivity-ub}]
The basic idea is to run BFS with an arbitrary starting vertex in $H$. A player $s\in \cK$ is determined as the start player and $s$ picks an arbitrary node in $H_s$ as the start vertex for the BFS.\footnote{To be completely correct, we have to make sure that $H_s$ is not empty. But this can be done by a simple leader election algorithm via  a Steiner tree style protocol (where each internal node passes on one of the incoming IDs to its parent and the ID picked by the root is declared the leader), which would be smaller than the bound we are after and hence, we should be able to ignore this.}

The idea is to simulate the BFS on $H$ in our framework. Let $D$ denote the diameter of $H$. 
We will use the flooding version of BFS. In particular, $s$ sends a {\em token} to all the neighbors of its chosen vertex in $H_s$. In future phases, every node in $H$ when it first receives a token, it sends the token to all of its neighbors. If the node has already received the token, then it just ignores the future receipt of the token.\footnote{The protocol needs to figure out a termination condition. By a simple Steiner tree type protocol one can count the number of nodes that have the token and we can stop if this number does not increase. To prevent overuse of this check, we can perform this in geometrically increasing round numbers.}

Consider the layered graph corresponding to the above run of the BFS on $H$. For layer $0\le i< D$, let $H_i$ denote the subgraph of $H$ that is involved in transfer of token when building layer $(i+1)$ from layer $i$. For notational simplicity let $n_i$ be the number of nodes in layer $i$, $m_i=|E(H_i)|$ and $\Delta_i$ denote the maximum degree of any node in layer $i$. Note since $H$ is randomly node distributed, then so is $H_i$.\footnote{All the bounds used in this proof hold with high enough probability so that we can apply union bound.}

Consider the case when we are building layer $(i+1)$ from layer $i$. Then the concentration bound proved in~\cite[Lemma 4.1]{klauck} implies that the corresponding multicommodity flow problem is for $\tO(m_i/k+\Delta_i)$-bounded demands.
This implies that we can simulate the BFS with
\begin{equation}
\label{eq:bfs-ub}
\sum_{i=0}^{D} \tO\left(\tau_{\MCF}\left(G,\cK, \frac{m_i}{k}+\Delta_i\right)\right)
\end{equation}

many rounds, where we have that
\begin{equation}
\label{eq:M-H-sum}
\sum_{i=0}^{D} m_i =\Theta(M_H),
\end{equation}
and
\begin{equation}
\label{eq:Delta-H-sum}
\sum_{i=0}^{D} \Delta_i =O(N_H).
\end{equation}

Now we assume that $H$ is large enough so that
\[\frac{M_H}{D\cdot k} \ge B_0,\]
where $B_0$ is as defined in Lemma~\ref{lem:MCF-large-dem}. Now note that for every $0\le i\le D$ such that $m_i < \frac{M_H}{D}$, we have that
\[\tau_{\MCF}\left(G,\cK, \frac{m_i}{k}+\Delta_i\right) \le \tau_{\MCF}\left(G,\cK, \frac{M_H}{D\cdot k}+\Delta_i\right).\]
Thus the total contribution of all such $i$ to the bound in~\eqref{eq:bfs-ub} is at most
\[\sum_{i=0}^D \tau_{\MCF}\left(G,\cK, \frac{M_H}{D\cdot k}+\Delta_i\right)   \le O\left(\tau_{\MCF}\left(G,\cK, \frac{M_H}{k}+N_H\right)\right),\]
where the inequality follows from Lemma~\ref{lem:MCF-large-dem} and~\eqref{eq:Delta-H-sum}.
Now for all $0\le i\le D$ such that $m_i\ge \frac{M_H}{D}$, from Claim~\ref{claim:tau-sub-additive}, we have that their contribution to~\eqref{eq:bfs-ub} is $O\left(\frac{m_i}{M_H}\cdot \tau_{\MCF}\left(G,\cK, \frac{M_H}{k}+\Delta_i\right)\right)$. Then by~\eqref{eq:M-H-sum} and~\eqref{eq:Delta-H-sum}, we have that the total contribution over all such $i$ is also $\tO\left(\tau_{\MCF}\left(G,\cK, \frac{M_H}{k}+N_H\right)\right)$. 

Thus, we have argued that~\eqref{eq:bfs-ub} is upper bounded by $\tO\left(\tau_{\MCF}\left(G,\cK, \frac{M_H}{k}+N_H\right)\right)$. If $H$ is large enough, then $M_H/k\ge N_H$, which would imply the claimed upper bound.
\end{proof}

Next we briefly state how we can use the standard extensions to BFS to extend the protocol in the proof above to work for other problems. To compute the connected components, change the above protocol so that when no more vertices are added to the current component, we check by the Steiner tree based leader election protocol to pick the next starting terminal $s$ and continue till we cannot. For the acyclicity problem, the above protocol should halt whenever a node receives the token more than once. Finally for bipartiteness, we pass two kinds of tokens: one for the odd rounds and one for the even rounds of the protocol and the graph is not bipartite if and only if a node received two different kinds of tokens. (For both the latter two modifications, we might also have to go through all connected components of $H$.) All this discussion implies that

\begin{thm}
\label{thm:other-graph-ub}
Let $H$ be a random node distributed graph. Then if $H$ is large enough compared to $G$, we can solve the connected components, acyclicity and bipartiteness problems on $H$ with $\tO\left(\tau_{\MCF}\left(G,\cK,\frac{M_H+ N_H}{k}\right)\right)$ randomized rounds of communication.
\end{thm}

\subsection*{Acknowledgments}

We would like to thank Jaikumar Radhakrishnan for helpful discussions at the early stage of this paper and to \href{https://simons.berkeley.edu/programs/inftheory2015}{Simons institute's Information theory program} for providing the venue for these discussions.


\newcommand{\etalchar}[1]{$^{#1}$}

\appendix

\section{A useful 2-party communication complexity result}
Consider the following 2-party problem. Alice (Bob) gets $k$ ($k'$) strings $\vx_1,\ldots, \vx_k$ ($\vy_1,\ldots,\vy_{k'}$), each $n$-bit long. They have to determine if one of Alice's strings is the same as that of one of Bob's, i.e. does there exist a pair $(i,j)$, such that $\vx_i = \vy_j$. (Note that this is same as checking whether the sets $\{\vx_1,\dots,\vx_k\}$ and $\{\vy_1,\dots,\vy_{k'}\}$ are disjoint.) Let us denote this problem as $\SD_n^{k , k'}$. 

\begin{thm}  \label{thm:EQ-k-by-k'}
The deterministic 2-party communication complexity of $\SD_n^{k, k'}$ is $\Omega\big(\min\set{k,k'}\cdot n\big)$ for $k,k' = 2^{o(n)}$.
\end{thm}

\begin{proof}
WLOG assume $k \le k'$. Pick some $t=k'-k$ strings from $\{0,1\}^n$. Let the set of remaining $2^n -t$  strings be called $T$. Alice and Bob each get $k$ strings from $T$ in the following way: partition $T$ into $k$ equal disjoint chunks, $T_1,\ldots,T_k$. Consider the problem where Alice and Bob each get $k$ strings, $\vx_1,\ldots,\vx_k$ and $\vy_1,\ldots,\vy_k$ respectively, with $\vx_i,\vy_i \in T_i$. They have to determine if for all $i$, $\vx_i \neq \vy_i$. Clearly if Alice and
Bob had a deterministic protocol of cost $c$ for solving $\SD_n^{k , k'}$, then they would also be able to solve this new problem $P$ with cost $c$ just as a special case. Note that $P$ is essentially $\text{AND}\circ \text{NEQ}$. The $i$-th NEQ instance has a Boolean matrix of dimension $|T_i| \times |T_i|$ whose rank is $|T_i| = \frac{2^n-t}{k}$. Then, $\text{AND} \circ \text{NEQ}$ matrix is the tensor product of these $k$ matrices. So its rank is $\left(\frac{2^n-t}{k}\right)^k$. Using the fact that communication
is lower bounded by the log of rank, we get that $c \ge k(\log(2^n- t) - \log k)$. The claim follows.
\end{proof} 

\end{document}